\documentclass[runningheads,envcountsame]{llncs}

\usepackage[T1]{fontenc}
\usepackage[utf8]{inputenc}
\usepackage{mathtools,amssymb,stmaryrd}
\usepackage{booktabs}
\usepackage{enumitem}
\usepackage{xcolor}
\usepackage{microtype}
\usepackage{hyperref}
\usepackage[nameinlink,noabbrev]{cleveref}

\setlist[itemize]{leftmargin=1.45em,itemsep=.08em,topsep=.18em}
\setlist[enumerate]{leftmargin=1.65em,itemsep=.10em,topsep=.20em}

\crefname{section}{Section}{Sections}
\Crefname{section}{Section}{Sections}
\crefname{subsection}{Section}{Sections}
\Crefname{subsection}{Section}{Sections}
\crefname{theorem}{Theorem}{Theorems}
\Crefname{theorem}{Theorem}{Theorems}
\crefname{lemma}{Lemma}{Lemmas}
\Crefname{lemma}{Lemma}{Lemmas}
\crefname{proposition}{Proposition}{Propositions}
\Crefname{proposition}{Proposition}{Propositions}
\crefname{corollary}{Corollary}{Corollaries}
\Crefname{corollary}{Corollary}{Corollaries}
\crefname{definition}{Definition}{Definitions}
\Crefname{definition}{Definition}{Definitions}
\crefname{remark}{Remark}{Remarks}
\Crefname{remark}{Remark}{Remarks}

\newcommand{\Set}{\mathbf{Set}}
\newcommand{\Setstar}{\mathbf{Set}_{*}}
\newcommand{\Vect}{\mathbf{Vect}}
\newcommand{\Vectop}{\mathbf{Vect}^{\mathrm{op}}}
\newcommand{\Cat}{\mathbf{Cat}}
\newcommand{\Syn}{\mathsf{Syn}}
\newcommand{\Synp}{\mathsf{Syn}_{+}}
\newcommand{\Tar}{\mathsf{Tar}}
\newcommand{\CTar}{\mathsf C_{\Tar}}
\newcommand{\LTar}{\mathsf L_{\Tar}}
\newcommand{\TotTar}{\mathsf T_{\Tar}^{\leftarrow}}
\newcommand{\Copow}{\mathsf{Copow}}
\newcommand{\Fam}{\mathsf{Fam}}
\newcommand{\Kar}{\mathsf{Kar}}
\newcommand{\Fix}{\operatorname{Fix}}
\newcommand{\im}{\operatorname{im}}

\newcommand{\Ddep}{\overleftarrow{\mathcal D}_{\mathrm{dep}}}
\newcommand{\DKar}{\overleftarrow{\mathcal D}_{\mathrm{Kar}}}
\newcommand{\Dst}{\overleftarrow{\mathcal D}_{\mathrm{st}}}
\newcommand{\TotTarFwd}{\mathsf T_{\Tar}^{\rightarrow}}
\newcommand{\DFKar}{\overrightarrow{\mathcal D}_{\mathrm{Kar}}}
\newcommand{\DFst}{\overrightarrow{\mathcal D}_{\mathrm{st}}}
\newcommand{\Rlz}{\mathfrak R}
\newcommand{\Prs}{\mathfrak P}
\newcommand{\semiarrow}{\rightsquigarrow}
\newcommand{\onehot}{\mathsf{one}}

\title{Simply Typed Reverse-Mode Automatic Differentiation with Variants: Denotational Correctness via Idempotent Completion}
\titlerunning{Simply Typed Reverse-Mode AD with Variants}
\author{Fernando Lucatelli Nunes\inst{1} \and Diogo Simm\inst{2} \and Matthijs V\'ak\'ar\inst{2}}
\authorrunning{F. Lucatelli Nunes, D. Simm, and M. V\'ak\'ar}
\institute{CMUC, Department of Mathematics, University of Coimbra, 3000-143 Coimbra, Portugal\\
\email{fln@uc.pt}
\and
Department of Information and Computing Sciences, Utrecht University, The Netherlands\\
\email{\{d.l.simmsallesvianna,m.i.l.vakar\}@uu.nl}}
\begin{document}
\emergencystretch=2em
\maketitle

\begin{abstract}
Reverse-mode automatic differentiation can be derived denotationally as a
structure-preserving interpretation of program syntax. In the usual simply
typed model, each source type has one cotangent type. Variants break this
representation because the valid cotangent space depends on the branch selected
at run time; established correctness results therefore use primal-indexed
families of cotangent spaces, whose direct internal language is dependent.

We show that this dependence can instead be represented in an ordinary
nondependent target. The cotangent fibres of each source type are placed inside
a common ambient type, and a primal-indexed idempotent selects the valid fibre.
For a category $\mathcal C$ and a regular infinite cardinal $\kappa$, we prove
that the constant-family inclusion extends to an equivalence
$$
\Kar\bigl(\Copow_\kappa(\mathcal C)\bigr)
\simeq
\Fam_\kappa(\mathcal C)
$$
precisely when $\mathcal C$ is Cauchy complete and every $\kappa$-small family
admits a common retract host. The same host condition yields explicit
coproducts after Karoubi completion.

We use this result to construct a bicartesian closed semantics for reverse-mode
AD with variants whose ambient types, projectors, and backpropagators are
ordinary nondependent target terms. Splitting the generated idempotents recovers
the dependent semantics, and the two interpretations are naturally isomorphic.
The ambient backpropagator is consequently the unique map that agrees with the
transposed derivative on the valid cotangent fibres and respects their
projectors.
\end{abstract}

\section{Introduction}
\label{sec:intro}

Reverse-mode automatic differentiation (AD) transforms a program into one that
propagates cotangents backward through its computation. Its efficiency for maps
with many inputs and few outputs makes it the standard basis of gradient
computation~\cite{GriewankWalther2008,BaydinEtAl2018,Margossian2019}. A
language-level account must make this propagation compositional in substitution,
binding, and the source type formers while proving that the transformed program
computes the transposed derivative.

Combinatory Homomorphic Automatic Differentiation (CHAD) obtains the
transformation from semantics~\cite{Vakar2021,VakarSmeding2022,LucatelliNunesVakar2023}.
Types and programs form a syntactic category, substitution is composition, and
the universal property of the syntax extends chosen primal maps and
codifferentials to a structure-preserving interpretation. Composition gives the
reverse chain rule; products and exponentials generate the corresponding
program clauses.

A concrete simply typed model for products and functions is
$\Copow(\Vectop)$. Its objects $(A,V)$ consist of a set of primal values and one
cotangent space, and an arrow $(f,f^*):(A,V)\to(B,W)$ has $f:A\to B$ and
$f_a^*:W\to V$, with $(gf)^*_a=f_a^*g^*_{f(a)}$. The category is cartesian
closed. Variants expose its limitation: at a value of $A+B$, the valid
cotangent is the left or right fibre according to the active branch, and
$\Copow(\Vectop)$ does not generally have coproducts. Expressive CHAD therefore
uses $\Fam(\Vectop)$, whose fibres vary with the primal value
~\cite{LucatelliNunesVakar2023}; Paszke and Plotkin use the analogous dependent
linear representation for sum types~\cite{PaszkePlotkin2023}. Executable CHAD uses common carriers
~\cite[Section~15.2 and Appendix~C]{LucatelliNunesVakar2023}, while Efficient
CHAD develops sparse representations~\cite{SmedingVakar2024}. This leaves open
why the simply typed representation presents the dependent semantics and what
determines its behaviour outside the valid fibres.

The key observation is that cotangent spaces $V$ and $W$ can be represented in
$V\oplus W$ by the projectors
$$
p_{\mathsf{inl}(a)}(v,w)=(v,0),
\qquad
p_{\mathsf{inr}(b)}(v,w)=(0,w).
$$
The dependency is retained as a primal-indexed family of idempotents. Since the
Karoubi envelope is the universal category in which idempotents split
~\cite{BorceuxDejean1986,Kelly1982,MacLane1998}, this suggests
$$
\Fam(\Vectop)
\simeq
\Kar\bigl(\Copow(\Vectop)\bigr)
\semiarrow
\Copow(\Vectop).
$$
The equivalence turns projector images into explicit fibres; the semifunctorial
passage retains the ordinary carriers and maps but no longer treats their
projectors as identities.

We first prove that common retract hosts generate $\kappa$-small coproducts in
$\Kar(\Copow_\kappa(\mathcal C))$ and characterize when the constant-family
inclusion extends to
$\Kar(\Copow_\kappa(\mathcal C))\simeq\Fam_\kappa(\mathcal C)$: precisely when
$\mathcal C$ is Cauchy complete and every $\kappa$-small family has a common
retract host. For $\Vectop$, direct sums provide such hosts. We then reproduce
the completion inside the ordinary CHAD target. For the freely generated
bicartesian closed syntax $\Synp$ and ordinary reverse target
$\TotTar=\Sigma_{\CTar}\LTar^{\mathrm{op}}$, retaining the generated
idempotents gives
$$
\DKar:\Synp\longrightarrow\Kar(\TotTar).
$$
Its representatives are ordinary nondependent types and terms. At data types
the idempotent is $(1,p_\tau)$; at function types cartesian closure projects the
entire function--backpropagator package by $h\mapsto e_\sigma h e_\tau$.
Forgetting the distinguished identities gives a composition-preserving
semifunctor $\Dst:\Synp\semiarrow\TotTar$.

Realization by fixed-point splittings is naturally isomorphic to dependent CHAD.
For a term $t:\tau\to\sigma$ between the data types of our language, the unique
projector-compatible ambient backpropagator is
$$
\widehat R_t(x)=\iota_{\tau,x}Df(x)^{\mathsf T}\rho_{\sigma,f(x)}.
$$

\paragraph{Contributions.}
\begin{enumerate}
\item We construct coproducts of completed constant families from common
retract hosts and prove the criterion
$\Kar(\Copow_\kappa(\mathcal C))\simeq\Fam_\kappa(\mathcal C)$.
\item We derive reverse AD with variants as a bicartesian closed interpretation
into the completed ordinary target, including explicit variant structure and
the higher-order idempotent forced by cartesian closure.
\item We prove equivalence with dependent CHAD, identify ordinary erasure as a
semifunctor, and derive the unique ambient correctness formula. The forward
dual is given in the appendix.
\end{enumerate}

\Cref{sec:basic} isolates the obstruction, \Cref{sec:completion} proves the
completion theorem, \Cref{sec:transformation} constructs the simply typed
interpretation, and \Cref{sec:correctness} proves correctness. The appendices
contain all expanded constructions and proofs.

\section{Semantics-driven reverse AD and the variant obstruction}
\label{sec:basic}

We begin by recalling the universal formulation of CHAD and the two fibrewise-dual target constructions. This separates the general semantics-driven principle from the specific obstruction introduced by variants: reverse differentiation itself remains compositional, but the constant-cotangent target lacks the coproducts needed to interpret branch-dependent fibres.

\subsection{Universal source semantics and fibrewise-dual targets}

Let $\Syn$ be the cartesian-closed syntactic category freely generated by ground types $\mathsf{real}^n$ and primitive operations, modulo the $\beta\eta$-equations and any declared equations among the primitives. Its universal property is the standard categorical semantics of the simply typed lambda calculus~\cite{Curien1986,LambekScott1986}: an interpretation of the generators satisfying those equations extends uniquely to a cartesian-closed functor.

Throughout, a \emph{bicartesian closed functor} means a functor equipped
with coherent comparison isomorphisms for the chosen initial object, terminal
object, binary products, binary coproducts, and exponentials. When the source
is the freely generated syntax and choices have been fixed, we suppress these
comparisons from the notation. No strictness assertion is intended.

The target language is organised by a cartesian category $\CTar$ of primal types and terms and an indexed category
$\LTar:\CTar^{\mathrm{op}}\longrightarrow\Cat$
of linear types and linear terms. Forward and reverse CHAD use the two fibrewise-dual Grothendieck constructions
$$
\TotTarFwd=\Sigma_{\CTar}\LTar,
\qquad
\TotTar=\Sigma_{\CTar}\LTar^{\mathrm{op}}.
$$
An arrow $(f,f^*):(A,V)\to(B,W)$ of $\TotTar$ consists of a primal term $f:A\to B$ and a term $f^*:A;W\to V$ linear in $W$; composition is substitution in the primal component and reverse accumulation in the linear component. We import the target hypotheses from basic CHAD: $(\CTar,\LTar)$ is a categorical model of its target language in the sense of~\cite[Definition~5.5]{VakarSmeding2022}. In particular, the forward and reverse total categories are cartesian closed by~\cite[Theorems~6.1 and~6.2]{VakarSmeding2022}. Broader closure criteria are given in~\cite{LucatelliNunesVakar2025Monoidal}.

Assigning each primitive operation its primal map and derivative therefore induces a forward functor $\overrightarrow{\mathcal D}:\Syn\to\TotTarFwd$, while assigning its primal map and codifferential induces the reverse functor
$$
\overleftarrow{\mathcal D}:\Syn\longrightarrow\TotTar.
$$
The two are fibrewise dual. We develop the reverse target below; every completion argument has the forward analogue obtained by removing the fibrewise opposite.

In the source language considered here, a \emph{data type} is generated from $\mathsf{real}^n$, $1$, and $0$ by finite products and variants. This is the tuple-and-variant data-type fragment of expressive CHAD\@. The correctness theorem in \Cref{sec:correctness} is stated for source terms whose domain and codomain are data types of this fragment.

This universal account works without variants because the constant-cotangent target has the cartesian closed structure required by products and functions. Coproducts are different: their tangent or cotangent fibre depends on the branch selected by the primal computation. We now make this obstruction precise in the reverse concrete model.

\subsection{The coproduct obstruction for constant cotangents}

The concrete constant-cotangent model is $\Copow(\Vectop)$. Its terminal object and binary products are
$$
1=(1,0),
\qquad
(A,V)\times(B,W)=(A\times B,V\oplus W).
$$

\begin{proposition}[Cartesian closure of the constant-cotangent model]
\label{prop:ccc}
The category $\Copow(\Vectop)$ is cartesian closed. An exponential may be chosen as
$$
(A,V)\Rightarrow(B,W)
=
\left(
A\to\bigl(B\times\Vect(W,V)\bigr),
\bigoplus_{a\in A}W
\right).
$$
\end{proposition}

\begin{proof}
This is the constant-cotangent instance of the cartesian-closed reverse target of basic CHAD~\cite{Vakar2021,VakarSmeding2022}. The displayed exponential follows by currying the primal map, separating the two biproduct components of the backward map, and applying the universal property of $\bigoplus_{a\in A}W$. The corresponding explicit evaluation and currying maps are recorded in \Cref{app:concrete-bcc}.
\end{proof}

The obstruction created by variants is now concrete.

\begin{proposition}[Coproduct obstruction]
	\label{prop:obstruction}
	The category $\Copow(\Vectop)$ does not have all binary coproducts.
\end{proposition}

\begin{proof}
	Suppose, for contradiction, that the coproduct $(K,U) = (1,0) \amalg (1,\mathbb{R})$ exists in $\Copow(\Vectop)$.
	
	Evaluating the covariant representable functor $\Copow(\Vectop)((K,U), -)$ on objects of the form $(B,0)$ for an arbitrary set $B$ yields the natural isomorphisms:
	$$
	\Set(K, B) \cong \Copow(\Vectop)((K,U),(B,0)) \cong B \times B \cong \Set(2, B).
	$$
	By the Yoneda lemma, this naturally forces $K \cong 2$.
	
	Next, we evaluate the same covariant representable functor against objects of the form $(1,T)$ for an arbitrary $T \in \Vect$. Since $K \cong 2$, the universal property of the coproduct provides the following natural bijections:
	$$
	\begin{aligned}
		\Vect(T,U\oplus U)
		&\cong \Copow(\Vectop)((K,U),(1,T)) \\
		&\cong
		\begin{multlined}[t]
		\Copow(\Vectop)((1,0),(1,T))\\
		{}\times\Copow(\Vectop)((1,\mathbb R),(1,T))
		\end{multlined}\\
		&\cong \Vect(T,0)\times\Vect(T,\mathbb R)
		\cong \Vect(T,\mathbb R).
	\end{aligned}
	$$
By the Yoneda lemma, it follows that $U \oplus U \cong \mathbb{R}$. This is impossible: since $\mathbb{R}$ is one-dimensional, such an isomorphism would imply that $U$ is finite-dimensional and that $2\dim(U)=1$.
\end{proof}

In $\Fam(\Vectop)$ the corresponding coproduct is branchwise:
$$
(A,V_{(-)})+(B,W_{(-)})=(A+B,Z_{(-)}),
$$
where $Z_{\mathsf{inl}(a)}=V_a$ and $Z_{\mathsf{inr}(b)}=W_b$. This is the cotangent semantics of variants in expressive CHAD~\cite{LucatelliNunesVakar2023}.

The obstruction is therefore not a failure of differentiation, but a mismatch between two free completions: basic CHAD admits only constant families, whereas variants require arbitrary families. The next section shows how common ambient hosts create the missing coproducts and when the resulting completion agrees with the family semantics.

\section{Coproducts and Cauchy completion of constant families}
\label{sec:completion}

The preceding obstruction compares two semantic universes. Constant families suffice for basic nondependent CHAD, while variants require fibres that vary with the selected primal branch. We now show that a common ambient object already supplies the missing coproducts after idempotents are split. Cauchy completeness then determines when this completed constant-family category agrees with the full family category.

Fix a regular infinite cardinal $\kappa$. A set is $\kappa$-small when its cardinality is less than $\kappa$. For the unadorned constructions used in the higher-order semantics, we fix Grothendieck universes $\mathcal U\in\mathcal V$: sets, families, and vector spaces are $\mathcal U$-small, and the resulting categories are regarded as $\mathcal V$-categories. Appendix~\ref{app:size} records the size conventions in full.

\subsection{Free coproduct, copower, and Cauchy completions}

For a category $\mathcal C$, the category $\Fam_\kappa(\mathcal C)$ has $\kappa$-small families $(X_i)_{i\in I}$ as objects. A morphism
$$
(u,\alpha):(X_i)_{i\in I}\longrightarrow(Y_j)_{j\in J}
$$
consists of a function $u:I\to J$ and arrows $\alpha_i:X_i\to Y_{u(i)}$. The singleton-family embedding
$$
\eta_{\mathcal C}:\mathcal C\longrightarrow\Fam_\kappa(\mathcal C)
$$
exhibits $\Fam_\kappa(\mathcal C)$ as the free completion of $\mathcal C$ under $\kappa$-small coproducts~\cite{Kelly1982,AdamekRosicky2020}.

The full subcategory on constant families is $\Copow_\kappa(\mathcal C)$. Its objects are pairs $(I,X)$, and it is the free completion of $\mathcal C$ under copowers by $\kappa$-small sets. We write
$$
\iota_{\mathcal C}:\mathcal C\longrightarrow\Copow_\kappa(\mathcal C)
$$
for the singleton embedding and
$$
J_{\mathcal C}:\Copow_\kappa(\mathcal C)\hookrightarrow\Fam_\kappa(\mathcal C)
$$
for the constant-family inclusion.

The Karoubi envelope $\Kar(\mathcal A)$ has objects $(A,e)$, where $e:A\to A$ is idempotent, and arrows $f:(A,e)\to(B,d)$ satisfying $f=dfe$. Its identity on $(A,e)$ is $e$. The canonical embedding $Q_{\mathcal A}:\mathcal A\to\Kar(\mathcal A)$ is the free Cauchy completion: for every Cauchy-complete category $\mathcal D$, restriction induces an equivalence
$$
[\Kar(\mathcal A),\mathcal D]
\simeq
[\mathcal A,\mathcal D].
$$
See~\cite{BorceuxDejean1986,Kelly1982,MacLane1998}; the universal properties and our variance conventions are recalled in the appendices.

The three constructions now have distinct roles. The copower completion provides the constant ambient carriers available to nondependent CHAD\@. The Karoubi envelope separates the subobjects selected by idempotents on those carriers. The family completion presents the resulting fibres directly. We first show that the common-host hypothesis already equips the Karoubi-completed constant-family category with the required coproducts.

\subsection{Coproducts in the Karoubi-completed constant-family category}

\begin{definition}
	\label{def:host}
	A family $(X_i)_{i\in I}$ in $\mathcal C$ \emph{admits a common retract host}
	if there are an object $P$ and, for each $i\in I$, arrows
	$$
	X_i\xrightarrow{s_i}P\xrightarrow{r_i}X_i
	$$
	such that
	$$
	r_is_i=1_{X_i}.
	$$
\end{definition}

A common retract host allows several ambient objects to be represented inside
a single carrier. We first isolate this change of carrier.

\begin{lemma}[Common-carrier presentation]
	\label{lem:common-carrier}
	Let $(E_a)_{a\in A}$ be a family of objects of
	$\Kar(\Copow_\kappa(\mathcal C))$, where
	$$
	E_a=\bigl((I_a,X_a),e_a\bigr),
	\qquad
	e_a=(u_a,p^a).
	$$
	Suppose that $(X_a)_{a\in A}$ admits a common retract host
	$$
	X_a\xrightarrow{s_a}P\xrightarrow{r_a}X_a,
	\qquad
	r_as_a=1_{X_a}.
	$$
	For each $a\in A$, define
	$$
	\widetilde p^a_i=s_ap^a_ir_a
	$$
	and
	$$
	\widetilde E_a
	=
	\bigl((I_a,P),\widetilde e_a\bigr),
	\qquad
	\widetilde e_a=(u_a,\widetilde p^a).
	$$
	Then $\widetilde e_a$ is idempotent, and there is an isomorphism
	$$
	E_a\cong\widetilde E_a
	$$
	in $\Kar(\Copow_\kappa(\mathcal C))$.
\end{lemma}

Once all objects have the same ambient carrier, their coproduct is obtained
by concatenating the index sets and assembling the idempotents blockwise.

\begin{proposition}[Coproducts from common retract hosts]
	\label{prop:kar-coprod}
	Suppose that every $\kappa$-small family in $\mathcal C$ admits a common
	retract host. Then
	$$
	\Kar\bigl(\Copow_\kappa(\mathcal C)\bigr)
	$$
	admits $\kappa$-small coproducts.
	
	More explicitly, let $(E_a)_{a\in A}$ be a $\kappa$-small family of objects,
	where
	$$
	E_a=\bigl((I_a,X_a),(u_a,p^a)\bigr),
	$$
	with
	$$
	u_a^2=u_a,
	\qquad
	p^a_{u_a(i)}p^a_i=p^a_i.
	$$
	Choose a common retract host
	$$
	X_a\xrightarrow{s_a}P\xrightarrow{r_a}X_a,
	\qquad
	r_as_a=1_{X_a},
	$$
	for the family $(X_a)_{a\in A}$. Put
	$$
	I=\coprod_{a\in A}I_a
	=
	\{(a,i)\mid a\in A,\ i\in I_a\},
	$$
	and define an endomorphism $c=(u,c_{(-)})$ of $(I,P)$ by
	$$
	u(a,i)=(a,u_a(i)),
	\qquad
	c_{(a,i)}=s_ap^a_ir_a.
	$$
	Then $c$ is idempotent, and a coproduct of $(E_a)_{a\in A}$ is
	$$
	\coprod_{a\in A}E_a
	=
	\bigl((I,P),c\bigr).
	$$
	If
	$$
	j_a:(I_a,X_a)\longrightarrow(I,P)
	$$
	has base map $i\mapsto(a,i)$ and constant component $s_a$, then the
	coproduct injection is
	$$
	\iota_a=c\,j_a\,(u_a,p^a).
	$$
	Equivalently, $\iota_a$ has base map
	$$
	i\longmapsto(a,u_a(i))
	$$
	and component $s_ap^a_i$ at $i$.
\end{proposition}

\begin{proof}
By \Cref{lem:common-carrier}, each $E_a$ is isomorphic to
$\widetilde E_a=((I_a,P),\widetilde e_a)$, with
$\widetilde p^a_i=s_ap^a_ir_a$. The displayed endomorphism $c$ is the
blockwise concatenation of these idempotents, so $c^2=c$, and the displayed
maps $\iota_a$ are Karoubi arrows.

Given Karoubi arrows $f_a:E_a\to D$, transport them along
$E_a\cong\widetilde E_a$. On the block $I_a$ of
$I=\coprod_a I_a$, define the base map and components of $[(f_a)]$ to be
those of the transported $f_a$. The sandwich equation holds blockwise, so
$[(f_a)]$ is the unique Karoubi arrow with $[(f_a)]\iota_a=f_a$. If
$A=\varnothing$, the common-host condition supplies an object $P$, and the
unique endomorphism of $(\varnothing,P)$ yields an initial Karoubi object.
Thus the construction gives all $\kappa$-small coproducts.
\end{proof}

The common-host hypothesis is used only to replace the ambient objects
$X_a$ by a single carrier $P$. Once this has been done, the coproduct is formed
by concatenating the index sets and assembling the corresponding idempotents
blockwise. 

A component-by-component verification, including the empty coproduct, is given in \Cref{app:common-carrier-details,app:coproduct-details}.

A useful sufficient condition is the existence of zero morphisms together
with the relevant products or coproducts. Indeed, if $\mathcal C$ has zero
morphisms and $\kappa$-small coproducts, then the coproduct of any
$\kappa$-small family is a common retract host: the coproduct injection
serves as the section, while the retraction is the morphism that is the
identity on the selected summand and zero on every other summand. More
generally, we have the following result.

\begin{lemma}[Common hosts from products or coproducts]
\label{lem:common-hosts-products-coproducts}
Let $\mathcal C$ be a category with zero morphisms, and let
$(X_i)_{i\in I}$ be a family of objects.

\begin{enumerate}
\item If the product
$$
P=\prod_{j\in I}X_j
$$
exists, then $P$ is a common retract host for $(X_i)_{i\in I}$.

\item If the coproduct
$$
Q=\coprod_{j\in I}X_j
$$
exists, then $Q$ is a common retract host for $(X_i)_{i\in I}$.
\end{enumerate}

Consequently, if $\mathcal C$ has zero morphisms and either
$\kappa$-small products or $\kappa$-small coproducts, then every
$\kappa$-small family in $\mathcal C$ admits a common retract host. Under
either hypothesis,
$$
\Kar\bigl(\Copow_\kappa(\mathcal C)\bigr)
$$
admits $\kappa$-small coproducts.
\end{lemma}

\begin{proof}
Suppose first that $P=\prod_{j\in I}X_j$ exists, with projections
$\pi_j:P\to X_j$. For each $i\in I$, let
$\delta_i:X_i\to P$ be the unique morphism such that
$$
\pi_j\delta_i=
\begin{cases}
1_{X_i},&j=i,\\
0_{X_i,X_j},&j\neq i.
\end{cases}
$$
Then $\pi_i\delta_i=1_{X_i}$, so $P$ is a common retract host, with
section $\delta_i$ and retraction $\pi_i$.

Dually, suppose that $Q=\coprod_{j\in I}X_j$ exists, with injections
$\iota_j:X_j\to Q$. For each $i\in I$, let $\rho_i:Q\to X_i$ be
the unique morphism determined by
$$
\rho_i\iota_j=
\begin{cases}
1_{X_i},&j=i,\\
0_{X_j,X_i},&j\neq i.
\end{cases}
$$
Then $\rho_i\iota_i=1_{X_i}$, so $Q$ is a common retract host, with
section $\iota_i$ and retraction $\rho_i$. The final assertion follows
from \Cref{prop:kar-coprod}.
\end{proof}

\subsection{Cauchy density and the family-completion equivalence}

The preceding proposition supplies the colimit structure required by variants. To identify the resulting category with $\Fam_\kappa(\mathcal C)$, one further condition is needed: 

\begin{lemma}[Cauchy completeness of family categories]
\label{lem:fam-cauchy}
The category $\Fam_\kappa(\mathcal C)$ is Cauchy complete if and only if
$\mathcal C$ is Cauchy complete.
\end{lemma}

\begin{proof}
Assume that $\mathcal C$ is Cauchy complete and let $(u,\alpha)$ be an
idempotent on $(X_i)_{i\in I}$. Its base idempotent splits through
$I_0=\Fix(u)$. For $j\in I_0$, split the idempotent $\alpha_j$ as
$X_j\xrightarrow{q_j}\overline X_j\xrightarrow{m_j}X_j$. The family maps
with base functions $i\mapsto u(i)$ and $I_0\hookrightarrow I$, and components
$q_{u(i)}\alpha_i$ and $m_j$, split $(u,\alpha)$.

Conversely, split an idempotent $e:X\to X$ after regarding it as an
idempotent on the singleton family $(X)$. The base equations force the
splitting family to be singleton, and its components split $e$ in
$\mathcal C$.
\end{proof}

\begin{lemma}[Retracts of constant families]
\label{lem:host}
A family is a retract in $\Fam_\kappa(\mathcal C)$ of a constant family if and only if it has a common retract host in $\mathcal C $.
\end{lemma}

\begin{proof}
If $(X_i)_{i\in I}$ has host $P$, the maps with base $1_I$ and components $s_i$ and $r_i$ exhibit it as a retract of the constant family $(P)_{i\in I}$. Conversely, suppose that $(X_i)_{i\in I}$ is a retract of $(P)_{j\in J}$, with section $(a,s)$ and retraction $(b,r)$. The equality $(b,r)(a,s)=1$ gives $ba=1_I$ and
$$
r_{a(i)}s_i=1_{X_i}.
$$
Thus $P$ is a common retract host.
\end{proof}

A fully faithful functor $J:\mathcal A\to\mathcal B$ is \emph{Cauchy dense} when every object of $\mathcal B$ is a retract of an object in the image of $J$. By \Cref{lem:host}:
\begin{lemma}
\label{lem:cauchy-dense}
 The inclusion $J_{\mathcal C}:\Copow_\kappa(\mathcal C)\hookrightarrow\Fam_\kappa(\mathcal C)$ is Cauchy dense if and only if every $\kappa$-small family has a common retract host. 
\end{lemma} 
This observation links the explicit coproducts of \Cref{prop:kar-coprod} with the universal comparison between the two completions.

\begin{theorem}[Family-completion criterion]
\label{thm:criterion}
For a category $\mathcal C$, the following are equivalent.
\begin{enumerate}
\item The constant-family inclusion extends to an equivalence
\begin{equation}\label{eq:fundamental-Karequivalence}
\Rlz:
\Kar\bigl(\Copow_\kappa(\mathcal C)\bigr)
\simeq
\Fam_\kappa(\mathcal C).
\end{equation}
\item The category $\mathcal C$ is Cauchy complete and every $\kappa$-small family in $\mathcal C$ has a common retract host.
\end{enumerate}
In $\textup{(1)}$, the equivalence is required to restrict, up to natural isomorphism, to $J_{\mathcal C}$. When these conditions hold, $\Rlz$ carries the coproducts of \Cref{prop:kar-coprod} to the concatenation coproducts in $\Fam_\kappa(\mathcal C)$.
\end{theorem}

\begin{proof}
Assume (2). By \Cref{lem:fam-cauchy}, $\Fam_\kappa(\mathcal C)$ is Cauchy
complete, so the fully faithful inclusion $J_{\mathcal C}$ extends along the
Karoubi embedding to
$R:\Kar(\Copow_\kappa(\mathcal C))\to\Fam_\kappa(\mathcal C)$.
By \Cref{lem:host,lem:cauchy-dense}, the common-host condition is
precisely Cauchy density of $J_{\mathcal C}$. Hence $R$ is fully faithful and
essentially surjective.

Conversely, assume (1). Since every Karoubi envelope is Cauchy complete,
$\Fam_\kappa(\mathcal C)$ is Cauchy complete, and \Cref{lem:fam-cauchy}
implies that $\mathcal C$ is Cauchy complete. Every $R(A,e)$ is a retract of
$R(A,1)\cong J_{\mathcal C}(A)$; essential surjectivity therefore makes
$J_{\mathcal C}$ Cauchy dense. By \Cref{lem:cauchy-dense}, every
$\kappa$-small family has a common retract host. Finally, the equivalence
carries the coproducts of \Cref{prop:kar-coprod} to concatenation coproducts
by uniqueness of coproducts. Detailed realization formulas are given in
\Cref{app:criterion-details,app:realization}.
\end{proof}

\begin{corollary}
\label{cor:vect}
Let $\mathcal C$ be a Cauchy-complete category with a zero object. If
$\mathcal C$ admits either $\kappa$-small products or $\kappa$-small
coproducts, then the constant-family inclusion extends to an equivalence
$$
\Kar\bigl(\Copow_\kappa(\mathcal C)\bigr)
\simeq
\Fam_\kappa(\mathcal C).
$$

\noindent In particular, there is an equivalence extending the constant-family inclusion
$$
\Rlz:
\Kar\bigl(\Copow_\kappa(\Vectop)\bigr)
\simeq
\Fam_\kappa(\Vectop).
$$
\end{corollary}

\begin{proof}
A zero object supplies zero morphisms. The first assertion therefore follows
from \Cref{lem:common-hosts-products-coproducts,thm:criterion}. The category
$\Vectop$ is Cauchy complete and has a zero object together with all small
products and coproducts. For the displayed ambient presentation, use the direct
sum $\bigoplus_iV_i$ in $\Vect$, which is a product in $\Vectop$: the
coordinate projection in $\Vect$ represents the section in $\Vectop$, and
the coordinate injection represents the retraction.
\end{proof}

Under this equivalence, a family $(V_i)_{i\in I}$ is represented by the
ambient space $\bigoplus_iV_i$ together with the coordinate projectors
$\iota_i\pi_i$. Conversely, splitting these projectors recovers the original
fibres. One could instead use direct-product ambient spaces; the two choices
give isomorphic coproduct objects after completion, although their ambient
presentations differ. We use direct sums because they provide the sparse,
finitely supported presentation relevant to reverse-mode implementations.

\section{A structure-preserving simply typed transformation}
\label{sec:transformation}

The completion theorem identifies the correct semantic target. We next lift that result from concrete families of vector spaces to the syntactic target language of CHAD\@. The construction proceeds in two stages: first we recover the missing bicartesian structure inside the Karoubi envelope; then the universal property of the source syntax generates the transformation.

\subsection{Variant coproducts in the completed target}

We use the ordinary reverse target $\TotTar$ of basic CHAD.

\paragraph{Variant-target assumption.}
We assume that $(\CTar,\LTar)$ is a categorical model of the basic CHAD
target language in the sense of~\cite[Definition~5.5]{VakarSmeding2022}.
Thus $\LTar$ is locally indexed, its fibres have the finite-biproduct and
cartesian linear-function structure required by CHAD, and $\TotTar$ is
cartesian closed by~\cite[Theorem~6.2]{VakarSmeding2022}. For variants, the
only additional assumption is that $\CTar$ has finite coproducts.

The remaining operations needed below are consequences of these hypotheses.
Finite biproducts in the fibres provide common retract hosts in the reverse
fibres. Moreover, cartesian linear-function objects and the coproduct universal
property in $\CTar$ give natural bijections
$$
\LTar(A+B)(U,P)
\cong
\LTar(A)(U,P)\times\LTar(B)(U,P),
$$
so linear maps over $A+B$ are uniquely determined by their two branch
restrictions. The associated branchwise linear case operation, including its
$\beta$- and $\eta$-equations and stability under reindexing, is therefore
derived rather than assumed. Appendix~\ref{app:target-closure} gives the full
argument.

The following calculations are written in the internal language of $\CTar$ and its indexed linear calculus. Equivalently, each displayed elementwise equation denotes the corresponding equality of target terms. We call an idempotent $e=(r,p)$ on a reverse-target object \emph{full} when it acts on the whole primal/function--backpropagator package, rather than only on a cotangent coordinate.

Let
$$
E=((A,V),e),
\qquad
F=((B,W),d),
$$
where $e=(r,p)$ and $d=(s,q)$ are idempotent endomorphisms in $\TotTar$. On the ambient object $(A+B,V\oplus W)$, define $c_{E,F}$ by
$$
\begin{aligned}
c_{E,F,1}(\mathsf{inl}(a))&=\mathsf{inl}(r(a)),
&
(c_{E,F})^*_{\mathsf{inl}(a)}(v,w)&=(p_a(v),0),\\
c_{E,F,1}(\mathsf{inr}(b))&=\mathsf{inr}(s(b)),
&
(c_{E,F})^*_{\mathsf{inr}(b)}(v,w)&=(0,q_b(w)).
\end{aligned}
$$
For an idempotent in the reverse target, the backward equations are $p_a\circ p_{r(a)}=p_a$ and $q_b\circ q_{s(b)}=q_b$. Hence $c_{E,F}^2=c_{E,F}$ branchwise. Define
$$
E+F:=((A+B,V\oplus W),c_{E,F}).
$$
This is the coproduct object in the completed category; no dependent target type former is required.

\begin{proposition}[Bicartesian closure of the completed target]
\label{prop:bcc}
Under the hypotheses above, $\Kar(\TotTar)$ is bicartesian closed. For
$E=((A,V),e)$ and $F=((B,W),d)$, its chosen structure is
$$
1=Q_{\TotTar}(1_{\TotTar}),
\quad
E\times F=((A\times B,V\oplus W),e\times d),
\quad
F^E=((B,W)^{(A,V)},\varepsilon_{e,d}),
$$
where the product projections are $e\pi_1,d\pi_2$ and
$\varepsilon_{e,d}=\Lambda(d\,\mathrm{ev}(1\times e))$, so
$\varepsilon_{e,d}(h)=dhe$. The initial object is $((0,0),1)$, and the binary coproduct is $((A+B,V\oplus W),c_{E,F})$.
\end{proposition}

\begin{proof}
The category $\TotTar$ is cartesian closed by the basic CHAD hypotheses, and its Karoubi completion inherits this structure by the standard Karoubi construction, recalled in \Cref{prop:kar-ccc-app}. The finite biproduct $V\oplus W$ is a common retract host for $V$ and $W$ in the reverse fibres, and the branchwise decomposition derived above gives
$[a,b]^*_{\mathsf{inl}(x)}(u)=(a_x^*(u),0)$ and
$[a,b]^*_{\mathsf{inr}(y)}(u)=(0,b_y^*(u))$.
The sandwich and injection equations prove existence, while the primal and linear $\eta$-equations determine both branches and hence prove uniqueness. The initial object is obtained from the initial object of $\CTar$ and the zero object in the linear fibres. Appendix~\ref{app:target-closure} contains the component proof and derives all branchwise operations from the stated hypotheses.
\end{proof}

Under \Cref{cor:vect}, this coproduct realizes to the branchwise coproduct in $\Fam(\Vectop)$. It is therefore an ambient presentation of the established dependent variant semantics, not a competing interpretation.

Having reconstructed the missing coproducts in the completed target, the rest of the transformation follows from the source universal property. This is the decisive methodological point: the clauses for variants are not appended to CHAD by hand, but are forced by the same preservation principle as the clauses for products and functions.

\subsection{The universal completed interpretation}

Let $\Synp$ be the bicartesian closed syntactic category freely generated by the source ground types and primitive operations, modulo the declared equations. Products, variants, and function types are supplied by its chosen bicartesian closed structure.

\begin{theorem}[Universal completed reverse interpretation]
\label{thm:semantic-kar}
Every assignment of the source ground types and primitive operations to objects and arrows of $\Kar(\TotTar)$ that respects the declared equations extends, relative to the chosen bicartesian closed structures, to a bicartesian closed functor
$$
\DKar:\Synp\longrightarrow\Kar(\TotTar).
$$
For the standard primal interpretations and primitive codifferentials, its ambient representatives are ordinary nondependent target types and terms. In particular, the variant clauses are generated by preservation of the completed coproducts rather than postulated as additional syntax-directed rules.
\end{theorem}

\begin{proof}
The substantive input is \Cref{prop:bcc}: common ambient hosts and idempotent splitting provide the coproducts that were absent from the ordinary reverse target, while Karoubi completion preserves its cartesian closed structure. Hence $\Kar(\TotTar)$ is bicartesian closed. The universal property of the freely generated bicartesian closed category $\Synp$ then extends the chosen interpretation of the generators to a bicartesian closed functor $\DKar$. For fixed choices of structure, this extension is determined by its action on the generators; different coherent choices are related by a structure-preserving natural isomorphism.

An object of $\Kar(\TotTar)$ is an ordinary target type together with an idempotent ordinary target term, and a morphism is an ordinary target term satisfying the corresponding sandwich equation. Thus the types, primal computations, backpropagators, and idempotents generated by $\DKar$ all belong to the ordinary nondependent target language. Products, variants, abstraction, application, and case analysis are consequently determined by preservation of the chosen bicartesian closed structure. The change from dependent fibres to ambient projectors changes the representation of cotangents, not the universal definition of the CHAD transformation.
\end{proof}

Concretely, for every source type $\tau$ write
$$
\DKar(\tau)
=
\bigl((\widehat D_1\tau,\widehat D_2\tau),e_\tau\bigr),
\qquad
e_\tau=(r_\tau,p_\tau),
$$
where $e_\tau$ is a full idempotent endomorphism of the ordinary reverse target object. For data types, induction gives $r_\tau=1$ and the linear component $p_\tau(x,-)$ is an idempotent cotangent projector. At ground types it is the identity. Products and sums act componentwise; in particular,
$$
\begin{aligned}
\widehat D_1(\tau+\sigma)
&=\widehat D_1\tau+\widehat D_1\sigma,\\
\widehat D_2(\tau+\sigma)
&=\widehat D_2\tau\oplus\widehat D_2\sigma,\\
p_{\tau+\sigma}(\mathsf{inl}(x),(v,w))
&=(p_\tau(x,v),0),\\
p_{\tau+\sigma}(\mathsf{inr}(y),(v,w))
&=(0,p_\sigma(y,w)).
\end{aligned}
$$
At a function type, however,
$$
e_{\tau\Rightarrow\sigma}
=
\varepsilon_{e_\tau,e_\sigma},
$$
the exponential idempotent induced by cartesian closure. Internally it sends an ambient function--backpropagator package $h$ to $e_\sigma h e_\tau$; its primal component need not be the identity. The idempotents are generated code, not proof objects. A transformed term is an ordinary primal computation and backpropagator satisfying the full sandwich equation at its source and target. Appendix~\ref{app:concrete-bcc} gives the concrete exponential idempotent explicitly.

The coproduct object, injections, and copairing above determine all four variant constructs---left and right injection, abort, and case analysis---without an ad hoc semantic clause. Appendix~\ref{app:macro} gives the resulting injection and case-analysis terms in full. Product and higher-order clauses are those generated by basic CHAD, restricted by the source and target idempotents; we do not duplicate its full macro here.

The fibrewise-dual forward construction and its ambient formula are recorded in Appendix~\ref{app:forward}.

The Karoubi-valued functor is the primary semantic construction. The ordinary nondependent program is obtained only after a deliberate change of viewpoint in which the ambient carriers and maps are retained but their distinguished idempotents are no longer treated as categorical identities.

\subsection{Idempotent erasure and semifunctoriality}

A semifunctor preserves objects, arrows, and composition, but need not preserve identities. Hoofman and Moerdijk showed, building on Hayashi and Hoofman, that the theory of semifunctors embeds fully into ordinary functor theory through Karoubi completion~\cite{Hayashi1985,Hoofman1993,HoofmanMoerdijk1995}. Forgetting the distinguished status of the idempotent attached to each Karoubi object gives such a map
$$
\Dst:\Synp\semiarrow\TotTar.
$$

\begin{proposition}[Idempotent erasure]
\label{prop:semi}
For every source arrow $f:\tau\to\sigma$,
$$
\Dst(f)=e_\sigma\circ\Dst(f)\circ e_\tau,
\qquad
\Dst(1_\tau)=e_\tau,
$$
and $\Dst(gf)=\Dst(g)\Dst(f)$. At data types $e_\tau=(1,p_\tau)$, so the backward component of the erased identity is the cotangent projector $p_\tau$.
\end{proposition}

\begin{proof}
Karoubi arrows satisfy the sandwich equation, their identities are the attached
idempotents, and their composition is ambient composition. The data-type claim
specializes this observation to $e_\tau=(1,p_\tau)$.
\end{proof}

Equivalently, every semifunctor $F:\mathcal X\semiarrow\mathcal A$ determines an ordinary functor $\mathcal X\to\Kar(\mathcal A)$ by sending $X$ to $(FX,F1_X)$. This is the concrete form, used here, of the Hoofman--Moerdijk passage from semifunctors to ordinary functors; the appendices prove the correspondence explicitly. Notice that this forgetful passage does not remove projector calls required by the generated program. It forgets only their distinguished categorical role as identities of completed objects.

Structure preservation explains how the algorithm is generated, but it does not by itself identify the linear map computed on the entire ambient carrier. The next section compares the completed semantics with dependent CHAD and derives the exact extension of the mathematical derivative to invalid ambient coordinates.

\section{Correctness of the ambient simply typed realization}
\label{sec:correctness}

We now compare the completed target with the concrete dependent model and use
the Karoubi equations to determine the ambient backpropagator.

\subsection{Concrete realization and comparison}

\paragraph{Concrete-model assumption.}
Assume a primal interpretation $\mathcal S:\CTar\to\Set$ and compatible
indexed linear interpretations
$\mathcal L_A:\LTar(A)\to\Vect^{\mathcal S A}$. A linear object $V$ is sent
to the constant family at a vector space $\mathcal V V$, and reindexing is
interpreted by substitution. Assume the usual basic-CHAD comparison maps, so
the induced reverse-total interpretation preserves the chosen terminal object,
products, and exponentials, and assume additionally that finite coproducts,
zero objects, biproducts, and linear-function objects are preserved. In
particular,
$$
\begin{aligned}
\mathcal S(A+B)&\cong\mathcal S A+\mathcal S B,\\
\mathcal V(V\oplus W)&\cong\mathcal V V\oplus\mathcal V W,\\
\mathcal S(V\multimap W)&\cong\Vect(\mathcal V V,\mathcal V W).
\end{aligned}
$$
Ground types have their usual Euclidean interpretations, and primitives use the
same primal maps and codifferentials as dependent CHAD. These data induce
$$
\llbracket-\rrbracket_\Tar:\TotTar\longrightarrow\Copow(\Vectop),
\qquad
\llbracket(A,V)\rrbracket_\Tar=(\mathcal S A,\mathcal V V),
$$
with
$\llbracket(f,\ell)\rrbracket_\Tar=
(\mathcal S f,(\mathcal L_A(\ell)_a)_{a\in\mathcal S A})$.
Karoubi completion followed by realization gives
$$
\llbracket-\rrbracket_{\mathrm{cpl}}:
\Kar(\TotTar)
\xrightarrow{\Kar(\llbracket-\rrbracket_\Tar)}
\Kar(\Copow(\Vectop))
\xrightarrow{\Rlz}
\Fam(\Vectop).
$$
Concretely, $\Rlz$ splits idempotents using their fixed-point sets in $\Set$
and their underlying fixed-point subspaces in $\Vect$.

\begin{proposition}[Structure preserved by concrete realization]
\label{prop:standard-target-compatibility}
The functor $\llbracket-\rrbracket_{\mathrm{cpl}}$ has coherent bicartesian
closed comparison isomorphisms. In particular, for $E,F\in\Kar(\TotTar)$,
$$
\llbracket E+F\rrbracket_{\mathrm{cpl}}
\cong
\llbracket E\rrbracket_{\mathrm{cpl}}
+
\llbracket F\rrbracket_{\mathrm{cpl}},
$$
where the right-hand side is the branchwise family coproduct. Under these
comparisons, ground types and primitives agree with dependent CHAD.
\end{proposition}

\begin{proof}
The inherited structure is preserved by the assumed concrete comparisons and
\Cref{thm:explicit-bcc-comparison}. For coproducts, realization splits the left
and right branch projectors of $c_{E,F}$, yielding precisely the two branchwise
image families. The full calculation is in Appendix~\ref{app:concrete-comparison}.
\end{proof}

Let $\Ddep:\Synp\to\Fam(\Vectop)$ denote the standard dependent reverse-CHAD
interpretation.

\begin{theorem}[Comparison with dependent CHAD]
\label{thm:comparison}
There is a coherent bicartesian closed natural isomorphism
$$
\llbracket-\rrbracket_{\mathrm{cpl}}\DKar
\cong
\Ddep.
$$
Thus explicit cotangent fibres and ambient carriers equipped with idempotents
present the same denotational transformation.
\end{theorem}

\begin{proof}
By \Cref{prop:standard-target-compatibility}, both sides are bicartesian closed
interpretations of $\Synp$ agreeing on ground types, primitives, and declared
equations. Freeness gives the unique coherent comparison. See
Appendix~\ref{app:concrete-comparison} for the expanded argument.
\end{proof}

Every data type generated from $\mathsf{real}^n$, $1$, and $0$ by finite
products and variants is isomorphic, by repeated distributivity, to a finite
sum of finite products of Euclidean ground types. The tuple-and-variant
correctness theorem of expressive CHAD therefore applies componentwise
~\cite[Theorem~124]{LucatelliNunesVakar2023}, and correctness transports along
these canonical structural isomorphisms.

\subsection{Projector-compatible ambient correctness}

For a data type $\tau$ and value $x$, let $T^*_\tau(x)$ be the genuine
cotangent fibre and $\widehat T^*_\tau$ its ambient cotangent space. Using
\Cref{thm:comparison}, identify $T^*_\tau(x)$ with the fixed-point subspace
$\Fix(p_{\tau,x})\subseteq\widehat T^*_\tau$, and use the resulting splitting
$$
T^*_\tau(x)
\xrightarrow{\iota_{\tau,x}}
\widehat T^*_\tau
\xrightarrow{\rho_{\tau,x}}
T^*_\tau(x),
\qquad
\rho_{\tau,x}\iota_{\tau,x}=1,
\quad
\iota_{\tau,x}\rho_{\tau,x}=p_{\tau,x}.
$$
An ambient extension $h$ of
$k:T^*_\sigma(f(x))\to T^*_\tau(x)$ is \emph{projector compatible} when
$h=p_{\tau,x}hp_{\sigma,f(x)}$ and
$\rho_{\tau,x}h\iota_{\sigma,f(x)}=k$.

\begin{lemma}[Unique projector-compatible extension]
\label{lem:unique-ambient-extension}
The unique projector-compatible extension of $k$ is
$$
h=\iota_{\tau,x}k\rho_{\sigma,f(x)}.
$$
\end{lemma}

\begin{proof}
The displayed map has the required restriction and sandwich equation. Conversely,
$h=p_\tau hp_\sigma=\iota_\tau(\rho_\tau h\iota_\sigma)\rho_\sigma$.
\end{proof}

\begin{theorem}[Ambient correctness for reverse CHAD]
\label{thm:ambient}
Let $t:\tau\to\sigma$ have domain and codomain among the data types generated
from $\mathsf{real}^n$, $1$, and $0$ by finite products and variants. Under the
usual differentiability and codifferential hypotheses on primitives, the
denotation $f=\llbracket t\rrbracket$ is componentwise differentiable and the
backward component of the erased nondependent transformation is
$$
\widehat R_t(x)
=
\iota_{\tau,x}Df(x)^{\mathsf T}\rho_{\sigma,f(x)}.
$$
It is the unique projector-compatible ambient extension of
$Df(x)^{\mathsf T}$.
\end{theorem}

\begin{proof}
Dependent correctness and \Cref{thm:comparison} identify the restriction of the
ambient map with $Df(x)^{\mathsf T}$; the Karoubi sandwich equation and
\Cref{lem:unique-ambient-extension} force the displayed formula. The expanded
calculation and a branch-sensitive example are in Appendix~\ref{app:correctness}.
\end{proof}

In particular, $\widehat R_{1_\tau}(x)=p_{\tau,x}$, while adjacent
retraction--section pairs cancel under composition.

\section{Related work and scope}

Basic and expressive CHAD provide, respectively, the higher-order
semantics-driven transformation and the dependent-family correctness result used
here~\cite{Vakar2021,VakarSmeding2022,LucatelliNunesVakar2023}. Efficient and
iterative CHAD study sparse accumulation and partial or iterative computation
~\cite{SmedingVakar2024,LucatelliNunesVakar2024ML,LucatelliNunesPlotkinVakar2025}.
For sums, Paszke and Plotkin likewise use dependent linear types selected by the
active branch~\cite{PaszkePlotkin2023}. Our contribution is the denotational
equivalence between such fibres and ordinary ambient types with idempotents,
including the reconstructed bicartesian closed structure and the uniquely
compatible ambient derivative; it is not a new asymptotic AD algorithm.

Other compositional accounts use lenses, reverse derivative categories,
continuations, control, or linear negation
~\cite{JohnsonRosebrughWood2012,FongSpivakTuyeras2019,CockettEtAl2020,CruttwellEtAl2022,PearlmutterSiskind2008,WangEtAl2019,BrunelMazzaPagani2020}.
Karoubi completion and semifunctors also occur in denotational and quantum
semantics~\cite{Scott1976,Hayashi1985,Hoofman1993,HoofmanMoerdijk1995,Selinger2008}.
Our correctness theorem concerns exactly finite products and variants of
Euclidean ground types.

\section{Conclusion}
\label{sec:conclusion}

Idempotent completion turns ambient simply typed cotangents into a bicartesian
closed model of variants equivalent to the dependent family semantics.
Realization recovers dependent CHAD, erasure yields its ordinary semifunctorial
program, and the sandwich equations force the unique ambient derivative
extension. The forward dual is given in the appendix.

\paragraph{Funding.}
This research was supported by the ERC project FoRECAST. The first author was
also supported by the Deutsche Forschungsgemeinschaft (DFG, German Research
Foundation) through the project Higher-Order Monad-based Programming and
Reasoning (HOMBRe), project number 501369690, led by Sergey Goncharov, during
his postdoctoral appointment at the School of Computer Science, University of
Birmingham. He was further supported by the Fields Institute for Research in
Mathematical Sciences through a Fields Research Fellowship in 2023, and by the
Centre for Mathematics of the University of Coimbra (CMUC) under the
Funda\c{c}\~ao para a Ci\^encia e a Tecnologia (FCT), through the grants
UID/00324/2025 and UID/PRR/00324/2025.

\paragraph{Data availability.}
This work is theoretical and does not use or generate research data.

\clearpage
\appendix
\crefname{section}{Appendix}{Appendices}
\Crefname{section}{Appendix}{Appendices}
\crefname{subsection}{Appendix}{Appendices}
\Crefname{subsection}{Appendix}{Appendices}
\renewcommand*{\theHsection}{app.\Alph{section}}
\renewcommand*{\theHsubsection}{app.\Alph{section}.\arabic{subsection}}
\section{Key ideas: from dependent fibres to simply typed programs}
\label{app:key-ideas}

This appendix gives a conceptual map of the construction for readers coming from
programming languages, category theory, or automatic differentiation.  It omits
most universal-property proofs, which begin in \Cref{app:size}, and concentrates
on what changes---and what does not---when dependent cotangent fibres are
presented by ordinary simply typed carriers.

\subsection{The semantic mismatch}

A reverse-mode value has two parts: a primal value and a space in which
cotangents for that value live.  For a source type without variants, one may use
one cotangent space for every primal value.  This gives the constant-family
model
$$
\Copow(\Vectop).
$$
A variant value is different.  At $\mathsf{inl}(a)$ only a left cotangent is
valid, while at $\mathsf{inr}(b)$ only a right cotangent is valid.  The intrinsic
semantic object is therefore a family
$$
x\longmapsto T^*_\tau(x),
$$
represented in $\Fam(\Vectop)$.  The dependence is essential: it expresses
which cotangents are meaningful at the current primal value.

The implementation problem is that a direct internal language for such a family
uses a cotangent type depending on the primal term.  Ordinary implementations,
however, prefer one target type for every source type.

\subsection{Keep the dependency as a projector}

Choose one ambient space $\widehat T^*_\tau$ large enough to contain every fibre
$T^*_\tau(x)$ as a retract.  For each primal value $x$, choose maps
$$
T^*_\tau(x)
\xrightarrow{\iota_{\tau,x}}
\widehat T^*_\tau
\xrightarrow{\rho_{\tau,x}}
T^*_\tau(x),
\qquad
\rho_{\tau,x}\iota_{\tau,x}=1.
$$
Their composite
$$
p_{\tau,x}=\iota_{\tau,x}\rho_{\tau,x}
$$
is an idempotent on the single ambient type.  Its image is the valid cotangent
fibre.  Thus a dependent family can be presented by two simply typed pieces of
data:
$$
\text{one ambient carrier}
\quad+\quad
\text{a primal-indexed projector}.
$$
For a sum, the ambient carrier is a biproduct and the projector selects the
active branch:
$$
\begin{aligned}
p_{\tau+\sigma}(\mathsf{inl}(x),(v,w))&=(p_\tau(x,v),0),\\
p_{\tau+\sigma}(\mathsf{inr}(y),(v,w))&=(0,p_\sigma(y,w)).
\end{aligned}
$$
Nothing dependent has been approximated or forgotten; its information has moved
from the target type into an idempotent target term.

\subsection{Why the Cauchy completion is the right category}

The ordinary constant-family category sees the ambient carrier but treats its
identity map as the identity.  The semantic fibre, however, behaves as if the
projector $p_x$ were the identity: valid values are fixed by $p_x$, and valid
maps satisfy sandwich equations.  The Karoubi envelope makes this viewpoint
literal.  An object is a pair $(A,e)$ with $e^2=e$, its identity is $e$, and an
arrow $h:(A,e)\to(B,d)$ satisfies
$$
h=dhe.
$$
Consequently,
$$
\Fam(\Vectop)
\simeq
\Kar\bigl(\Copow(\Vectop)\bigr)
$$
expresses that an explicit family of fibres and an ambient carrier with
projectors are two presentations of the same semantic object.  The relevant
completion is $\Kar(\Copow(\Vectop))$, not
$\Kar(\Fam(\Vectop))$: the latter is already equivalent to
$\Fam(\Vectop)$ because the family category is Cauchy complete.

The general theorem isolates the precise hypotheses.  The constant-family
inclusion extends along idempotent completion to an equivalence with all
$\kappa$-small families if and only if $\mathcal C$ is Cauchy complete and every
such family has a common retract host.
For vector spaces, the host is a direct sum and every idempotent splits through
its image.

\subsection{One transformation, three views}

The construction has three compatible levels:
$$
\begin{array}{ccl}
\Fam(\Vectop)
&\simeq&
\Kar(\Copow(\Vectop))
\quad\semiarrow\quad
\Copow(\Vectop)\\[2mm]
\text{dependent fibres}
&&
\text{ambient carrier + identities}
\qquad
\text{ambient carrier only.}
\end{array}
$$
The left presentation is best for stating the mathematical semantics and proving
logical-relations correctness.  The middle presentation is best for preserving
products, variants, and higher-order functions: it is an ordinary category in
which the projectors are the identities of the represented fibres.  The right
presentation is the executable simply typed program obtained by forgetting only
that distinguished categorical role.  The corresponding idempotent calls remain in the code.

This final passage is a semifunctor rather than a functor.  Composition is still
strict, because composition in the Karoubi envelope is ambient composition.
Identities become the generated idempotents:
$$
\Dst(1_\tau)=e_\tau.
$$
At data types, $e_\tau=(1,p_\tau)$, so the primal identity is unchanged
while the cotangent identity discards invalid coordinates.

\subsection{Variants and higher-order functions}

For variants, completion creates the missing coproduct.  Its ordinary ambient
type is
$$
(A+B,V\oplus W),
$$
but its identity is the branch-sensitive idempotent.  The coproduct injections
and case analysis are obtained by sandwiching the ordinary terms with the source
and target idempotents.  This is why the resulting clauses satisfy the
coproduct equations at the completed level even though the ordinary
constant-cotangent category lacks coproducts.

At higher-order types, projecting only the cotangent coordinate is insufficient:
an ambient function--backpropagator package may itself fail to respect the
represented input and output fibres.  Cartesian closure generates the required
full idempotent.  Internally it sends
$$
h\longmapsto e_\sigma h e_\tau.
$$
Its fixed points are precisely the projector-compatible function packages.
Restricted evaluation and currying then give the higher-order translation.  No
separate ad hoc rule is needed.

\subsection{Why the ambient derivative is forced}

Let $h:\widehat T^*_\sigma\to\widehat T^*_\tau$ be an ambient backpropagator.
Correctness on the genuine fibres requires
$$
\rho_{\tau,x}h\iota_{\sigma,f(x)}=Df(x)^{\mathsf T},
$$
while being a completed arrow requires
$$
h=p_{\tau,x}hp_{\sigma,f(x)}.
$$
Substituting $p=\iota\rho$ leaves no freedom:
$$
h
=
\iota_{\tau,x}Df(x)^{\mathsf T}\rho_{\sigma,f(x)}.
$$
The formula explains the whole program: project the incoming ambient cotangent
to the valid output fibre, apply the genuine transposed derivative, and embed
the result into the ambient input carrier.

It also explains the apparently asymmetric functorial laws.  For an identity,
$Df(x)^{\mathsf T}=1$ and hence $h=p_{\tau,x}$.  For a composite, the middle
retraction--section pair cancels:
$$
\rho_{\sigma,f(x)}\iota_{\sigma,f(x)}=1,
$$
so the ordinary chain rule yields strict preservation of composition.

\subsection{Semantics versus representation}

The theorem does not prescribe a dense run-time representation.  An ambient
direct sum is the semantic vector space in which linear accumulation is defined;
an implementation may realise it sparsely, use tagged carriers, or apply the
accumulation techniques of Efficient CHAD\@.  The completion result supplies the
missing semantic invariant: whatever representation is chosen must implement
the projectors and the sandwich equations.  In this sense, dependent and simply
typed reverse AD are not competing algorithms.  They are the same
semantics-driven transformation viewed through explicit fibres, completed
ambient carriers, and executable ordinary terms.

\section{Size conventions and notation}
\label{app:size}

Fix a regular infinite cardinal $\kappa$. A set is $\kappa$-small when its cardinality is less than $\kappa$. The subscripted constructions $\Fam_\kappa$ and $\Copow_\kappa$ use $\kappa$-small index sets; regularity ensures that a $\kappa$-small coproduct of $\kappa$-small sets is again $\kappa$-small. These assumptions are all that the completion theorem requires. For the unadorned higher-order semantic categories $\Fam$ and $\Copow$, we fix a Grothendieck universe of small sets large enough to contain the relevant function sets and direct sums. Thus the cartesian-closed constructions below are universe-small without imposing an unsupported closure property on an arbitrary regular infinite cardinal. The finite version is obtained by replacing the $\kappa$-small doctrine with finite sets throughout. Our conventions for free completions, enriched copowers, and Cauchy completion follow the standard accounts in~\cite{Kelly1982,MacLane1998,BorceuxDejean1986}.

We write $\Vect$ for real vector spaces and linear maps, $\Vectop$ for its opposite, and $\Setstar$ for pointed sets and basepoint-preserving maps. Our section--retraction notation is always internal to the category under discussion. In particular, in $\Vectop$ a section is represented by a projection in $\Vect$, and a retraction by an inclusion in $\Vect$. For categories $\mathcal A$ and $\mathcal D$, the functor category is denoted $[\mathcal A,\mathcal D]$.

\section{Free completions by families and copowers}
\label{app:free-completions}

We begin with the two free completions whose difference creates the variant obstruction. Arbitrary families freely add coproducts, while constant families freely add only set-indexed copowers. These constructions are standard instances of free cocompletion~\cite{Kelly1982,AdamekRosicky2020}; family categories are also the semantic mechanism used for expressive CHAD and arise naturally in familial descent~\cite{LucatelliNunesVakar2023,Prezado2024}. Making both universal properties explicit fixes the variance and coherence conventions used throughout the later completion argument.

\subsection{The free coproduct completion}

\begin{definition}[Family category]
\label{def:fam-app}
The category $\Fam_\kappa(\mathcal C)$ has the following data.
\begin{itemize}
\item An object is a family $(X_i)_{i\in I}$, where $I$ is $\kappa$-small and $X_i\in\mathcal C$.
\item A morphism
$$
(u,\alpha):(X_i)_{i\in I}\longrightarrow(Y_j)_{j\in J}
$$
consists of a function $u:I\to J$ and arrows $\alpha_i:X_i\to Y_{u(i)}$.
\item Identities and composition are
$$
1_{(X_i)}=(1_I,(1_{X_i})_{i\in I}),
$$
$$
(v,\beta)\circ(u,\alpha)
=
(vu,(\beta_{u(i)}\circ\alpha_i)_{i\in I}).
$$
\end{itemize}
The singleton-family embedding is
$$
\eta_\mathcal C:\mathcal C\longrightarrow\Fam_\kappa(\mathcal C),
\qquad
X\longmapsto(X)_{*\in1}.
$$
\end{definition}

If $(X_i^a)_{i\in I_a}$ is a $\kappa$-small family of objects of $\Fam_\kappa(\mathcal C)$ indexed by $a\in A$, then its coproduct is concatenation:
$$
\coprod_{a\in A}(X_i^a)_{i\in I_a}
=
(X_i^a)_{(a,i)\in\coprod_{a\in A}I_a}.
$$
The injection from the $a$-th summand uses the inclusion $I_a\to\coprod_b I_b$ and identity component arrows.

The following universal property is the standard free completion by families; see, for example,~\cite[Section~5.3]{Kelly1982} and~\cite[Section~2]{AdamekRosicky2020}.

\begin{theorem}[Free $\kappa$-small coproduct completion]
\label{thm:fam-up-app}
Let $\mathcal D$ have chosen $\kappa$-small coproducts. Restriction along $\eta_\mathcal C$ induces an equivalence
$$
\eta_\mathcal C^*:
\operatorname{Fun}_{\amalg_\kappa}
\bigl(\Fam_\kappa(\mathcal C),\mathcal D\bigr)
\simeq
\operatorname{Fun}(\mathcal C,\mathcal D),
$$
where the left-hand side contains the functors preserving the chosen coproducts and all natural transformations between them.
\end{theorem}

\begin{proof}
Given $F:\mathcal C\to\mathcal D$, define
$$
\overline F((X_i)_{i\in I})
=
\coprod_{i\in I}F(X_i).
$$
For $(u,\alpha):(X_i)\to(Y_j)$, define $\overline F(u,\alpha)$ by the coproduct universal property: its restriction to the $i$-th summand is
$$
F(X_i)
\xrightarrow{F(\alpha_i)}
F(Y_{u(i)})
\xrightarrow{\iota_{u(i)}}
\coprod_{j\in J}F(Y_j).
$$
The restrictions of $\overline F(1)$ and of the identity agree on every summand, so they are equal. For composable $(u,\alpha)$ and $(v,\beta)$, the restrictions of $\overline F(v,\beta)\overline F(u,\alpha)$ and $\overline F(vu,\beta\alpha)$ to the $i$-th summand are both
$$
F(X_i)
\xrightarrow{F(\alpha_i)}
F(Y_{u(i)})
\xrightarrow{F(\beta_{u(i)})}
F(Z_{v(u(i))})
\xrightarrow{\iota_{v(u(i))}}
\coprod_kF(Z_k).
$$
Hence $\overline F$ is a functor. It preserves the chosen coproducts because concatenating the indexing sets corresponds to reassociating the chosen coproducts in $\mathcal D$.

A natural transformation $\theta:F\Rightarrow G$ extends to
$$
\overline\theta_{(X_i)}=\coprod_i\theta_{X_i}.
$$
Naturality follows by checking each coproduct injection. This defines a functor from the right-hand functor category to the left-hand one.

Conversely, let $H:\Fam_\kappa(\mathcal C)\to\mathcal D$ preserve the chosen coproducts. Every family is canonically a coproduct of singleton families:
$$
(X_i)_{i\in I}
\cong
\coprod_{i\in I}\eta_\mathcal C(X_i).
$$
Therefore
$$
H((X_i)_{i\in I})
\cong
\coprod_{i\in I}H(\eta_\mathcal C X_i),
$$
naturally in the family. Thus $H$ is naturally isomorphic to the extension of $H\eta_\mathcal C$. The same decomposition shows that a natural transformation between coproduct-preserving functors is uniquely determined by its components on singleton families. Restriction and extension are therefore mutually inverse equivalences.
\end{proof}

The family construction is the standard free coproduct completion used in the semantics of expressive CHAD~\cite{LucatelliNunesVakar2023}; see also the general free-cocompletion accounts in~\cite{Kelly1982,AdamekRosicky2020} and the use of family categories in familial descent~\cite{Prezado2024}.

The constant-family subcategory retains the same morphism calculus but suppresses variation in the fibre object. Its universal property is correspondingly weaker: it freely supplies copowers, rather than arbitrary coproducts.

\subsection{The free copower completion}

\begin{definition}[Constant-family category]
\label{def:copow-app}
The category $\Copow_\kappa(\mathcal C)$ is the full subcategory of $\Fam_\kappa(\mathcal C)$ on constant families. An object is written $(I,X)$. A morphism
$$
(u,\alpha):(I,X)\longrightarrow(J,Y)
$$
consists of a function $u:I\to J$ and arrows $\alpha_i:X\to Y$ for $i\in I$.
\end{definition}

The singleton embedding is
$$
\iota_\mathcal C:\mathcal C\longrightarrow\Copow_\kappa(\mathcal C),
\qquad
X\longmapsto(1,X).
$$
For a $\kappa$-small set $S$, define
$$
S\cdot(I,X)=(S\times I,X).
$$
There is a natural bijection
$$
\Copow_\kappa(\mathcal C)
\bigl(S\cdot(I,X),(J,Y)\bigr)
\cong
\Set\Bigl(S,\Copow_\kappa(\mathcal C)((I,X),(J,Y))\Bigr),
$$
obtained by currying the base function $S\times I\to J$ and the $S\times I$-indexed component arrows. Thus $S\cdot(I,X)$ is the copower by $S$.

The analogous copower completion is standard enriched-category theory; see~\cite[Sections~3.7 and~4.1]{Kelly1982}.

\begin{theorem}[Free completion under $\kappa$-small copowers]
\label{thm:copow-up-app}
Let $\mathcal D$ have chosen copowers by $\kappa$-small sets. Restriction along $\iota_\mathcal C$ induces an equivalence
$$
\iota_\mathcal C^*:
\operatorname{Fun}_{\odot_\kappa}
\bigl(\Copow_\kappa(\mathcal C),\mathcal D\bigr)
\simeq
\operatorname{Fun}(\mathcal C,\mathcal D).
$$
\end{theorem}

\begin{proof}
Given $F:\mathcal C\to\mathcal D$, define
$$
\overline F(I,X)=I\cdot F(X).
$$
For $(u,\alpha_i):(I,X)\to(J,Y)$, define $\overline F(u,\alpha)$ by copairing, over $i\in I$, the maps
$$
F(X)
\xrightarrow{F(\alpha_i)}
F(Y)
\xrightarrow{\iota_{u(i)}}
J\cdot F(Y).
$$
Identity and composition follow by checking the copower injections. The canonical associativity isomorphism
$$
S\cdot(I\cdot F(X))
\cong
(S\times I)\cdot F(X)
$$
shows that $\overline F$ preserves the chosen copowers.

Conversely, every object $(I,X)$ is the copower $I\cdot(1,X)$. Hence a copower-preserving functor is naturally isomorphic to the extension of its restriction to singleton objects. Natural transformations are likewise determined by their singleton components. This proves the equivalence.
\end{proof}

\section{Karoubi completion and semifunctors}
\label{app:kar-semi}

The two free completions above add different colimit structure. The Karoubi envelope performs a different operation: it adds the retracts already encoded by idempotents. We use the classical universal description of Cauchy completion~\cite{Kelly1982,BorceuxDejean1986,MacLane1998} and the semifunctorial perspective developed in programming-language semantics by Hayashi, Hoofman, and Hoofman--Moerdijk~\cite{Hayashi1985,Hoofman1993,HoofmanMoerdijk1995}. This is precisely the operation that turns an ambient cotangent carrier with projectors into its family of genuine fibres.

\subsection{The Karoubi envelope}

\begin{definition}[Karoubi envelope]
\label{def:kar-app}
For a category $\mathcal A$, the objects of $\Kar(\mathcal A)$ are pairs $(A,e)$ with $e:A\to A$ idempotent. A morphism
$$
f:(A,e)\longrightarrow(B,d)
$$
is an arrow $f:A\to B$ in $\mathcal A$ satisfying
$$
f=dfe.
$$
The identity on $(A,e)$ is $e$, and composition is inherited from $\mathcal A$.
\end{definition}

The condition $f=dfe$ is equivalent to the pair of equations $f=fe=df$. The canonical embedding is
$$
Q_\mathcal A:\mathcal A\longrightarrow\Kar(\mathcal A),
\qquad
A\longmapsto(A,1_A).
$$
Every $(A,e)$ is a retract of $Q_\mathcal A(A)$: both the section and the retraction are represented by $e$, and their composite on $(A,e)$ is $e$, the identity of that Karoubi object.

A category is Cauchy complete, or idempotent complete, when every idempotent splits. For ordinary categories this is equivalent to admitting all absolute colimits~\cite{BorceuxDejean1986,Kelly1982,MacLane1998}.

The following universal property is the classical description of Cauchy completion by splitting idempotents~\cite{BorceuxDejean1986,Kelly1982,MacLane1998}.

\begin{theorem}[Universal property of the Karoubi envelope]
\label{thm:kar-up-app}
If $\mathcal D$ is Cauchy complete, restriction along $Q_\mathcal A$ induces an equivalence
$$
Q_\mathcal A^*:
[\Kar(\mathcal A),\mathcal D]
\simeq
[\mathcal A,\mathcal D].
$$
\end{theorem}

\begin{proof}
Let $F:\mathcal A\to\mathcal D$. For every idempotent $e:A\to A$, choose a splitting
$$
F(A)
\xrightarrow{q_e}
\overline F(A,e)
\xrightarrow{m_e}
F(A),
$$
with $m_eq_e=F(e)$ and $q_em_e=1$. Define
$$
\overline F(f)=q_dF(f)m_e
$$
for $f:(A,e)\to(B,d)$. The identity of $(A,e)$ is sent to
$$
q_eF(e)m_e
=
q_em_eq_em_e
=1.
$$
For composable $f:(A,e)\to(B,d)$ and $g:(B,d)\to(C,c)$,
$$
\begin{aligned}
\overline F(g)\overline F(f)
&=q_cF(g)m_dq_dF(f)m_e\\
&=q_cF(g)F(d)F(f)m_e\\
&=q_cF(gf)m_e\\
&=\overline F(gf),
\end{aligned}
$$
because $g=gd$ and $f=df$. Thus $\overline F$ is a functor. On objects $(A,1_A)$ choose the trivial splitting, so $\overline FQ_\mathcal A=F$.

If $\theta:F\Rightarrow G$, define
$$
\overline\theta_{(A,e)}=q_e^G\theta_A m_e^F.
$$
Naturality follows by inserting the sandwich equations and the naturality of $\theta$. Conversely, a natural transformation between extensions is determined by its components on $Q_\mathcal A(A)$ because every $(A,e)$ is a retract of such an object. Different choices of splittings produce naturally isomorphic extensions, through the unique isomorphisms commuting with the splitting maps. Hence restriction is an equivalence.
\end{proof}

\subsection{Semifunctors}

A \emph{semifunctor} $F:\mathcal X\semiarrow\mathcal A$ assigns objects and arrows and preserves composition, but is not required to preserve identities~\cite{Hoofman1993,HoofmanMoerdijk1995}. For every object $X$, put $e_X=F(1_X)$. Then
$$
e_X^2=F(1_X1_X)=e_X,
$$
and for $f:X\to Y$,
$$
F(f)=F(1_Yf1_X)=e_YF(f)e_X.
$$
Thus $F(f)$ is a Karoubi arrow $(FX,e_X)\to(FY,e_Y)$.

The correspondence below is the object-and-arrow form of the embedding of semifunctor theory into ordinary functor theory developed by Hayashi, Hoofman, and Hoofman--Moerdijk~\cite{Hayashi1985,Hoofman1993,HoofmanMoerdijk1995}.

\begin{proposition}[Semifunctor--Karoubi correspondence]
\label{prop:semi-kar-app}
There is a bijective correspondence on object-and-arrow assignments
$$
\left\{\text{semifunctors }\mathcal X\semiarrow\mathcal A\right\}
\longleftrightarrow
\left\{\text{functors }\mathcal X\to\Kar(\mathcal A)\right\}.
$$
It sends $F$ to the functor
$$
\widetilde F(X)=(FX,F1_X),
\qquad
\widetilde F(f)=F(f),
$$
and sends a Karoubi-valued functor to its underlying semifunctor.
\end{proposition}

\begin{proof}
The equations above show that $\widetilde F(f)$ is a valid Karoubi arrow. It preserves identities because the identity of $(FX,F1_X)$ is $F1_X$, and it preserves composition because $F$ does. Conversely, forgetting the distinguished idempotents of a Karoubi-valued functor leaves a composition-preserving assignment. The two constructions are inverse.
\end{proof}

Taking as transformations the families $\theta_X:FX\to GX$ satisfying naturality and
$$
\theta_X=G1_X\,\theta_X\,F1_X
$$
yields an equivalence of categories between semifunctors with these transformations and ordinary functors into the Karoubi envelope.

\section{Coproducts, Cauchy completion, and ambient representations}
\label{app:completion-details}
\label{app:explicit-equivalence}

This appendix gives a self-contained proof of the categorical result used in
\Cref{sec:completion}.  The argument has three steps.  First, the common-host
condition produces explicit $\kappa$-small coproducts in the Karoubi envelope
of the constant-family category.  Second, the same condition is identified
with Cauchy density of the constant-family inclusion.  Finally, Cauchy density
and idempotent completeness yield the equivalence with the full family
category.  We conclude by recording the concrete passage between explicit
families and ambient carriers equipped with idempotents.

Throughout, $\kappa$ is a regular infinite cardinal and $\mathcal C$ is a
category.  The category $\Fam_\kappa(\mathcal C)$ has $\kappa$-small families
$(X_i)_{i\in I}$ as objects.  A morphism
$$
(u,\alpha):(X_i)_{i\in I}\longrightarrow(Y_j)_{j\in J}
$$
consists of a function $u:I\to J$ and arrows
$\alpha_i:X_i\to Y_{u(i)}$.  The full subcategory on the constant families is
$\Copow_\kappa(\mathcal C)$; its objects are written $(I,X)$.  We denote the
constant-family inclusion by
$$
J_{\mathcal C}:\Copow_\kappa(\mathcal C)
\hookrightarrow\Fam_\kappa(\mathcal C).
$$
An object of $\Kar(\Copow_\kappa(\mathcal C))$ is written
$$
E=((I,X),(u,p)),
$$
where $u^2=u$ and $p_{u(i)}p_i=p_i$ for every $i\in I$.

\subsection{Coproducts from common retract hosts}
\label{app:common-carrier-details}

Recall that a family $(X_a)_{a\in A}$ has a \emph{common retract host} if
there are an object $P$ and arrows
$$
X_a\xrightarrow{s_a}P\xrightarrow{r_a}X_a,
\qquad
r_as_a=1_{X_a},
$$
for every $a\in A$.

\begin{lemma}[Changing to a common carrier]
\label{lem:common-carrier-app}
Let
$$
E_a=((I_a,X_a),(u_a,p^a))
$$
be objects of $\Kar(\Copow_\kappa(\mathcal C))$, and suppose that the
objects $X_a$ have a common retract host $P$.  Put
$$
\widetilde p^a_i=s_ap^a_ir_a,
\qquad
\widetilde E_a=((I_a,P),(u_a,\widetilde p^a)).
$$
Then $(u_a,\widetilde p^a)$ is idempotent and
$E_a\cong\widetilde E_a$ in
$\Kar(\Copow_\kappa(\mathcal C))$.
\end{lemma}

\begin{proof}
Using $r_as_a=1$ and $p^a_{u_a(i)}p^a_i=p^a_i$, we obtain
$$
\widetilde p^a_{u_a(i)}\widetilde p^a_i
=
s_ap^a_{u_a(i)}r_as_ap^a_ir_a
=
s_ap^a_ir_a
=
\widetilde p^a_i.
$$
Hence $(u_a,\widetilde p^a)$ is idempotent.  Define
$$
\alpha_a=
\bigl(u_a,(s_ap^a_i)_{i\in I_a}\bigr):E_a\longrightarrow\widetilde E_a,
\qquad
\beta_a=
\bigl(u_a,(p^a_ir_a)_{i\in I_a}\bigr):
\widetilde E_a\longrightarrow E_a.
$$
The identities
$$
\widetilde p^a_{u_a(i)}s_ap^a_i=s_ap^a_i,
\qquad
p^a_{u_a(i)}r_a\widetilde p^a_i=p^a_ir_a
$$
show that these are Karoubi arrows.  Their composites have base map $u_a$
and components
$$
p^a_{u_a(i)}r_as_ap^a_i=p^a_i,
\qquad
s_ap^a_{u_a(i)}p^a_ir_a=\widetilde p^a_i.
$$
Thus $\beta_a\alpha_a=(u_a,p^a)$ and
$\alpha_a\beta_a=(u_a,\widetilde p^a)$, which are the identities of the two
Karoubi objects.
\end{proof}

\subsection{The blockwise coproduct}
\label{app:coproduct-details}

\begin{proposition}[$\kappa$-small coproducts from common hosts]
\label{prop:kar-coprod-app}
If every $\kappa$-small family in $\mathcal C$ admits a common retract host,
then $\Kar(\Copow_\kappa(\mathcal C))$ has $\kappa$-small coproducts.

More explicitly, let
$$
E_a=((I_a,X_a),(u_a,p^a)),
\qquad a\in A,
$$
be a $\kappa$-small family, and choose a common host
$X_a\xrightarrow{s_a}P\xrightarrow{r_a}X_a$.  Put
$$
I=\coprod_{a\in A}I_a
$$
and define $c=(u,c_{(-)})$ on $(I,P)$ by
$$
u(a,i)=(a,u_a(i)),
\qquad
c_{(a,i)}=s_ap^a_ir_a.
$$
Then $c$ is idempotent and
$$
\bigl((I,P),c\bigr)
$$
is a coproduct of the objects $E_a$.  Its $a$-th injection has base map
$i\mapsto(a,u_a(i))$ and component $s_ap^a_i$ at $i$.
\end{proposition}

\begin{proof}
Regularity of $\kappa$ ensures that $I$ is $\kappa$-small.  By
\Cref{lem:common-carrier-app}, it is enough to construct a coproduct of
$$
\widetilde E_a=((I_a,P),(u_a,\widetilde p^a)),
\qquad
\widetilde p^a_i=s_ap^a_ir_a.
$$
The equations
$$
u^2(a,i)=(a,u_a^2(i))=u(a,i),
\qquad
c_{u(a,i)}c_{(a,i)}
=
\widetilde p^a_{u_a(i)}\widetilde p^a_i
=
c_{(a,i)}
$$
show that $c$ is idempotent.

For each $a$, let
$$
\widetilde\iota_a:\widetilde E_a\longrightarrow((I,P),c)
$$
have base map $i\mapsto(a,u_a(i))$ and component $\widetilde p^a_i$.  This
is the sandwich of the ordinary block injection by the source and target
idempotents, and hence is a Karoubi arrow.

Let $H=((J,Y),(w,q))$, and suppose that
$$
f_a=(v_a,\phi^a):\widetilde E_a\longrightarrow H
$$
are Karoubi arrows.  Their sandwich equations are equivalent to
$$
v_au_a=v_a,
\qquad
\phi^a_{u_a(i)}\widetilde p^a_i=\phi^a_i,
\qquad
wv_a=v_a,
\qquad
q_{v_a(i)}\phi^a_i=\phi^a_i.
$$
Define $(v,\phi):((I,P),c)\to H$ blockwise by
$$
v(a,i)=v_a(i),
\qquad
\phi_{(a,i)}=\phi^a_i.
$$
The four displayed equations give $(v,\phi)c=(v,\phi)$ and
$(w,q)(v,\phi)=(v,\phi)$, so this is a Karoubi arrow.  Moreover,
$$
(v,\phi)\widetilde\iota_a=f_a
$$
for every $a$.

For uniqueness, let $g=(z,\gamma)$ be another such arrow.  Since $g=gc$,
$$
z(a,i)=z(a,u_a(i)),
\qquad
\gamma_{(a,i)}=
\gamma_{(a,u_a(i))}\widetilde p^a_i.
$$
The equation $g\widetilde\iota_a=f_a$ then forces
$z(a,i)=v_a(i)$ and $\gamma_{(a,i)}=\phi^a_i$.  Thus $g=(v,\phi)$.
Transporting the injections along the isomorphisms of
\Cref{lem:common-carrier-app} gives the stated component $s_ap^a_i$.

When $A=\varnothing$, the common-host condition supplies an object $P$.
The unique endomorphism of $(\varnothing,P)$ is idempotent, and the resulting
Karoubi object has a unique arrow to every object.  It is therefore initial.
\end{proof}

The construction has a simple interpretation.  The host condition is used
only to place all ambient objects over one carrier $P$; the coproduct then
concatenates the index sets and places the corresponding idempotents on
disjoint blocks.  Different choices of host may give different ambient
representatives, but the resulting coproducts are canonically isomorphic by
the universal property.

\subsection{Cauchy density and the completion theorem}
\label{app:criterion-details}

We first isolate the standard Cauchy-density argument used below.

\begin{lemma}[Cauchy-dense extensions]
\label{lem:cauchy-dense-extension-app}
Let $J:\mathcal A\to\mathcal D$ be fully faithful and let $\mathcal D$ be
Cauchy complete.  Let
$$
\overline J:\Kar(\mathcal A)\longrightarrow\mathcal D
$$
be an extension of $J$ along the canonical embedding
$Q_{\mathcal A}:\mathcal A\to\Kar(\mathcal A)$.  Then $\overline J$ is an
equivalence if and only if $J$ is Cauchy dense.
\end{lemma}

\begin{proof}
Choose $\overline J(A,e)$ as a splitting object of $J(e)$:
$$
J(A)\xrightarrow{q_e}\overline J(A,e)
\xrightarrow{m_e}J(A),
\qquad
m_eq_e=J(e),
\quad
q_em_e=1.
$$
For Karoubi objects $(A,e)$ and $(B,d)$, every arrow
$h:\overline J(A,e)\to\overline J(B,d)$ determines the ambient arrow
$m_dhq_e:J(A)\to J(B)$.  By full faithfulness, it is $J(f)$ for a unique
$f:A\to B$.  The splitting equations imply $f=dfe$, and
$q_dJ(f)m_e=h$.  Hence $\overline J$ is fully faithful.

If $J$ is Cauchy dense, every object of $\mathcal D$ is a retract of some
$J(A)$.  The corresponding idempotent on $J(A)$ is $J(e)$ for a unique
idempotent $e$ on $A$, and the retract is isomorphic to
$\overline J(A,e)$.  Thus $\overline J$ is essentially surjective.
Conversely, every $\overline J(A,e)$ is a retract of $J(A)$; essential
surjectivity of $\overline J$ therefore implies Cauchy density of $J$.
\end{proof}

\begin{lemma}[Cauchy completeness of family categories]
\label{lem:fam-cauchy-app}
The category $\Fam_\kappa(\mathcal C)$ is Cauchy complete if and only if
$\mathcal C$ is Cauchy complete.
\end{lemma}

\begin{proof}
Assume first that $\mathcal C$ is Cauchy complete.  Let $(u,\alpha)$ be an
idempotent on $(X_i)_{i\in I}$.  Put $I_0=\Fix(u)$.  For every $j\in I_0$,
choose a splitting
$$
X_j\xrightarrow{q_j}\overline X_j\xrightarrow{m_j}X_j
$$
of the idempotent $\alpha_j$.  Define
$$
q:(X_i)_{i\in I}\longrightarrow(\overline X_j)_{j\in I_0}
$$
to have base map $i\mapsto u(i)$ and component $q_{u(i)}\alpha_i$, and
define
$$
m:(\overline X_j)_{j\in I_0}\longrightarrow(X_i)_{i\in I}
$$
to have the inclusion $I_0\hookrightarrow I$ as base map and component
$m_j$.  Then $qm=1$ and $mq=(u,\alpha)$, so every idempotent splits.

Conversely, let $e:X\to X$ be idempotent and split the induced idempotent on
the singleton family $(X)$ through a family indexed by $J$.  The base map of
the composite from that family to itself is constant, whereas the composite
is the identity.  Hence $J$ is a singleton, and the component maps split $e$
in $\mathcal C$.
\end{proof}

\begin{proposition}[Cauchy completion commutes with families]
\label{prop:kar-fam-app}
For every category $\mathcal C$, there is an equivalence
$$
\Kar\bigl(\Fam_\kappa(\mathcal C)\bigr)
\simeq
\Fam_\kappa\bigl(\Kar(\mathcal C)\bigr)
$$
whose restriction along $Q_{\Fam_\kappa(\mathcal C)}$ is naturally
isomorphic to $\Fam_\kappa(Q_{\mathcal C})$.
\end{proposition}

\begin{proof}
The functor
$$
\Fam_\kappa(Q_{\mathcal C}):
\Fam_\kappa(\mathcal C)
\longrightarrow
\Fam_\kappa(\Kar(\mathcal C))
$$
is fully faithful.  Its codomain is Cauchy complete by
\Cref{lem:fam-cauchy-app}, since $\Kar(\mathcal C)$ is Cauchy complete.
It is Cauchy dense because each family of Karoubi objects is, componentwise,
a retract of the corresponding family of objects in the image of
$Q_{\mathcal C}$.  The result now follows from
\Cref{lem:cauchy-dense-extension-app}.
\end{proof}

\begin{lemma}[Retracts of constant families]
\label{lem:retract-constant-app}
A family $(X_i)_{i\in I}$ is a retract in $\Fam_\kappa(\mathcal C)$ of a
constant family if and only if it admits a common retract host in
$\mathcal C$.
\end{lemma}

\begin{proof}
If $P$ is a common host, the family maps with base $1_I$ and components
$s_i$ and $r_i$ exhibit $(X_i)_{i\in I}$ as a retract of $(P)_{i\in I}$.
Conversely, suppose that $(X_i)_{i\in I}$ is a retract of $(P)_{j\in J}$,
with section $(a,s)$ and retraction $(b,r)$.  From
$(b,r)(a,s)=1$ we obtain $ba=1_I$ and
$$
r_{a(i)}s_i=1_{X_i}.
$$
Thus $P$ is a common retract host.
\end{proof}

Since $J_{\mathcal C}$ is fully faithful, the preceding lemma immediately
gives
$$
J_{\mathcal C}\text{ is Cauchy dense}
\quad\Longleftrightarrow\quad
\text{every $\kappa$-small family in $\mathcal C$ has a common host}.
$$

\paragraph{Detailed proof of \Cref{thm:criterion}.}
Assume (2).  By \Cref{lem:fam-cauchy-app},
$\Fam_\kappa(\mathcal C)$ is Cauchy complete.  By
\Cref{lem:retract-constant-app}, the fully faithful inclusion
$J_{\mathcal C}$ is Cauchy dense.  Its extension along
$$
Q:\Copow_\kappa(\mathcal C)
\longrightarrow
\Kar(\Copow_\kappa(\mathcal C))
$$
therefore exists by the universal property of the Karoubi envelope and is an
equivalence by \Cref{lem:cauchy-dense-extension-app}.  This gives (1).
The existence of $\kappa$-small coproducts in the source of $\Rlz$ was proved
independently in \Cref{prop:kar-coprod}.

Conversely, assume (1).  The category
$\Kar(\Copow_\kappa(\mathcal C))$ is Cauchy complete, so its equivalent
category $\Fam_\kappa(\mathcal C)$ is Cauchy complete.  By
\Cref{lem:fam-cauchy-app}, $\mathcal C$ is Cauchy complete.  Moreover,
$Q$ is Cauchy dense and $\Rlz Q\cong J_{\mathcal C}$; hence
$J_{\mathcal C}$ is Cauchy dense.  By
\Cref{lem:retract-constant-app}, every $\kappa$-small family has a common
retract host.

It remains to identify the chosen coproducts.  For
$i\in\Fix(u_a)$, choose a splitting
$$
X_a\xrightarrow{q^a_i}X^a_i\xrightarrow{m^a_i}X_a
$$
of $p^a_i$.  The block idempotent $s_ap^a_ir_a$ is then split by
$$
P\xrightarrow{q^a_ir_a}X^a_i
\xrightarrow{s_am^a_i}P.
$$
Thus realization sends the block indexed by $(a,i)$ to $X^a_i$.  The
realized component of the injection is
$$
(q^a_ir_a)(s_ap^a_i)m^a_i
=
q^a_ip^a_im^a_i
=
1_{X^a_i}.
$$
Hence the realized object is the concatenation of the realized families and
the realized injections are the ordinary family coproduct injections.

\subsection{Explicit realization and presentation}
\label{app:realization}

\paragraph{Fixed-point splittings.}
In $\Set$ and $\Vect$ we split an idempotent $e:X\to X$ through $\Fix(e)$,
using $q_e(x)=e(x)$ and the inclusion $m_e$:
$X\xrightarrow{q_e}\Fix(e)\xrightarrow{m_e}X$. Thus $m_eq_e=e$ and
$q_em_e=1$. In any category, for a splitting
$X\xrightarrow{q}E\xrightarrow{m}X$, the map $m$ is the equalizer and $q$ the
coequalizer of $e$ and $1_X$; conversely, either construction, or any epi--mono
image or coimage factorization $e=mq$, yields a splitting. A second splitting
$X\xrightarrow{q'}E'\xrightarrow{m'}X$ is canonically isomorphic to the first
via $q'm:E\cong E'$, with inverse $qm'$. For $\Vectop$ we use the opposite of
this fixed-point splitting.

The equivalence has a useful concrete description.  Let
$$
E=((I,P),(u,p))
$$
be a completed constant family.  For each $i\in I_0:=\Fix(u)$, choose a
splitting
$$
P\xrightarrow{q_i}P_i\xrightarrow{m_i}P,
\qquad
m_iq_i=p_i,
\quad
q_im_i=1.
$$
Then
$$
\Rlz(E)=(P_i)_{i\in I_0}.
$$
If $(f,\phi):E\to F$ is a Karoubi arrow, then $f$ restricts to a function
$\Fix(u)\to\Fix(v)$ and the realized component at $i$ is
$$
q^F_{f(i)}\phi_i m^E_i.
$$
Thus realization replaces each ambient projector by its splitting object and
restricts each ambient arrow to the corresponding images.

Conversely, choose for every family $X=(X_i)_{i\in I}$ a common host
$$
X_i\xrightarrow{s_i^X}P_X\xrightarrow{r_i^X}X_i.
$$
Its ambient presentation is
$$
\Prs(X)=
\left(
(I,P_X),
\bigl(1_I,(s_i^Xr_i^X)_{i\in I}\bigr)
\right).
$$
A family arrow $(f,\alpha):X\to Y$ is represented by
$$
\Prs(f,\alpha)
=
\left(
 f,
 (s^Y_{f(i)}\alpha_i r_i^X)_{i\in I}
\right).
$$
The retraction equations make these formulas functorial: the middle factors
cancel under composition, and the ambient representative of an identity is
the corresponding projector.

\begin{proposition}[Explicit realization equivalence]
\label{prop:explicit-equivalence-app}
Under the hypotheses of \Cref{thm:criterion}, the formulas above define
inverse equivalences up to natural isomorphism
$$
\Rlz:
\Kar\bigl(\Copow_\kappa(\mathcal C)\bigr)
\rightleftarrows
\Fam_\kappa(\mathcal C):\Prs.
$$
For compatible choices of splittings, $\Rlz\Prs=1$ strictly.
\end{proposition}

\begin{proof}
Choose the splitting of $s_i^Xr_i^X$ to be
$$
P_X\xrightarrow{r_i^X}X_i\xrightarrow{s_i^X}P_X.
$$
Then $\Rlz\Prs$ is the identity on objects and arrows.  For
$E=((I,P),(u,p))$, present $\Rlz(E)$ using the same carrier $P$ and the
splittings of the $p_i$.  This gives
$$
\Prs\Rlz(E)=
\left((I_0,P),(1_{I_0},(p_i)_{i\in I_0})\right).
$$
The arrows with base maps
$$
I\xrightarrow{u}I_0,
\qquad
I_0\hookrightarrow I,
$$
and component $p_i$ in the corresponding block are inverse Karoubi arrows
between $E$ and $\Prs\Rlz(E)$.  Naturality follows from the sandwich equation
for a Karoubi arrow.  Hence $\Prs\Rlz\cong1$.
\end{proof}

In the concrete case $\mathcal C=\Vectop$, a family $(V_i)_{i\in I}$ may be
presented on the single ambient vector space $\bigoplus_iV_i$, with the
coordinate projectors selecting the valid fibre.  Splitting those projectors
recovers the original family.  This is the precise sense in which a dependent
cotangent family and a simply typed ambient cotangent equipped with
primal-indexed projectors are two representations of the same object.

\section{Common retract hosts from standard categorical structure}
\label{app:instances}

The common-host hypothesis of \Cref{thm:criterion} is often automatic. This
appendix records a general sufficient condition and then specializes it to the
examples relevant to the paper. The point is that a host need not be a
particular coproduct: either a product or a coproduct can serve, provided that
zero morphisms are available.

\subsection{General criterion}

\begin{proposition}[Standard sources of common retract hosts]
\label{prop:standard-hosts-app}
Let $\mathcal C$ be a category with zero morphisms.
\begin{enumerate}
\item If $\mathcal C$ has $\kappa$-small products, then every
$\kappa$-small family in $\mathcal C$ admits a common retract host.
\item If $\mathcal C$ has $\kappa$-small coproducts, then every
$\kappa$-small family in $\mathcal C$ admits a common retract host.
\end{enumerate}
Consequently,
$$
\Kar\bigl(\Copow_\kappa(\mathcal C)\bigr)
$$
has $\kappa$-small coproducts under either hypothesis. If, in addition,
$\mathcal C$ is Cauchy complete, then the constant-family inclusion extends
to an equivalence
$$
\Kar\bigl(\Copow_\kappa(\mathcal C)\bigr)
\simeq
\Fam_\kappa(\mathcal C).
$$
The same conclusions hold whenever $\mathcal C$ has a zero object, since a
zero object induces zero morphisms.
\end{proposition}

This is the combination of
\Cref{lem:common-hosts-products-coproducts,prop:kar-coprod,thm:criterion}.
For a product host, the section into the $i$-th coordinate has identity as its
$i$-th component and zero as every other component. For a coproduct host, the
retraction from the $i$-th summand is identity on that summand and zero on all
others. Different choices of hosts generally give different ambient
presentations, but the resulting coproduct objects in the Karoubi completion
are isomorphic by the coproduct universal property.

\subsection{Vector spaces and their opposite}

The categories $\Vect$ and $\Vectop$ are Cauchy complete, have zero objects,
and admit all small products and coproducts. Hence
$$
\Kar\bigl(\Copow_\kappa(\Vect)\bigr)
\simeq
\Fam_\kappa(\Vect)
$$
and
$$
\Kar\bigl(\Copow_\kappa(\Vectop)\bigr)
\simeq
\Fam_\kappa(\Vectop).
$$

For a family $(V_i)_{i\in I}$ in $\Vectop$, either of the following may be
used as a common host.
\begin{enumerate}
\item The direct sum $P_\oplus=\bigoplus_iV_i$ in $\Vect$, regarded as
a product in $\Vectop$. The coordinate projection
$\pi_i:P_\oplus\to V_i$ in $\Vect$ represents the section
$V_i\to P_\oplus$ in $\Vectop$, while the coordinate injection
$\iota_i:V_i\to P_\oplus$ represents the retraction.
\item The direct product $P_\Pi=\prod_iV_i$ in $\Vect$, regarded as a
coproduct in $\Vectop$. The product projection represents the coproduct
injection in $\Vectop$, while the one-coordinate map
$\delta_i:V_i\to P_\Pi$ represents its retraction.
\end{enumerate}
For finite families these hosts agree up to the canonical biproduct
isomorphism. For infinite families they generally differ. We use direct sums
because they give the sparse, finitely supported ambient representation. With
this choice, a family arrow whose backward components are
$\alpha_i^*:Y_{f(i)}\to X_i$ is represented on the $i$-th block by
$$
P_Y
\xrightarrow{\pi_{f(i)}}
Y_{f(i)}
\xrightarrow{\alpha_i^*}
X_i
\xrightarrow{\iota_i}
P_X,
$$
namely by $\iota_i\alpha_i^*\pi_{f(i)}$.

\subsection{Modules and additive categories}

Let $R$ be a ring. The categories $R\text{-}\mathbf{Mod}$ and
$(R\text{-}\mathbf{Mod})^{\mathrm{op}}$ are Cauchy complete, have zero
objects, and admit all small products and coproducts. Therefore the same two
completion equivalences hold for modules. More generally, any Cauchy-complete
category with zero morphisms and the required products or coproducts satisfies
\Cref{prop:standard-hosts-app}. In particular, this applies to
Cauchy-complete additive categories whenever the relevant products or
coproducts exist; finite biproducts automatically provide hosts for finite
families.

\subsection{Pointed sets and sparse tagged carriers}

The category $\Setstar$ of pointed sets is Cauchy complete, has a zero object,
and admits all small products and coproducts. Thus
$$
\Kar\bigl(\Copow_\kappa(\Setstar^{\mathrm{op}})\bigr)
\simeq
\Fam_\kappa(\Setstar^{\mathrm{op}}).
$$
A convenient host for a family of pointed sets $(X_i,*)$ is the wedge
$$
P=\bigvee_{i\in I}X_i.
$$
The inclusion of the $i$-th summand has a pointed retraction that is the
identity on that summand and collapses every other summand to the basepoint.
After passage to the opposite category, these maps give the required section
and retraction. For two branches, the wedge has the shape ``zero, a left
payload, or a right payload'', and the projectors collapse the unselected
branch to zero. This explains the carrier shape of a sparse tagged
implementation, although pointed sets themselves do not supply the linear
accumulation or transposed derivatives required for cotangent semantics.

\section{Explicit structure of the completed constant-cotangent model}
\label{app:concrete-bcc}

This appendix gives the bicartesian closed structure explicitly. Expressive
CHAD proves that $\Fam(\Vectop)$ is bicartesian closed and identifies it as
the concrete reverse-mode target~\cite[Section~10.2, Corollary~82, and
Equation~(68)]{LucatelliNunesVakar2023}. Together with
\Cref{cor:vect}, this already gives a bicartesian closed structure on
$\Kar(\Copow(\Vectop))$ by transport across realization. The purpose of
this appendix is to describe the terminal object, products, the initial object,
coproducts, exponentials, evaluation, and currying directly on ambient
constant-cotangent objects and to identify their images under realization.

Put
$$
\mathcal K=\Kar(\Copow(\Vectop)),
\qquad
\mathcal F=\Fam(\Vectop).
$$

\begin{theorem}[Bicartesian closed realization comparison]
\label{thm:explicit-bcc-comparison}
The realization equivalence
$$
\Rlz:\mathcal K\simeq\mathcal F
$$
can be equipped with coherent comparison isomorphisms for the terminal object,
binary products, the initial object, arbitrary coproducts, and exponentials.
Consequently, it is an equivalence of bicartesian closed categories.
\end{theorem}

\begin{proof}
Existence follows by transporting the standard bicartesian closed structure of
$\mathcal F$ across the equivalence. The remainder of this appendix gives the
chosen ambient structure and the comparison isomorphisms explicitly and
verifies their compatibility with evaluation and currying.
\end{proof}

We now give the explicit ambient account.

\subsection{Objects, arrows, and realization}

An object of $\mathcal K$ is written
$$
E=((A,V),e),
\qquad
e=(r,p),
$$
where $r:A\to A$ and $p_a:V\to V$ satisfy
$$
r^2=r,
\qquad
p_a p_{r(a)}=p_a.
$$
Let
$$
F=((B,W),d),
\qquad
d=(s,q).
$$
A morphism $h:E\to F$ is represented by a function $h_0:A\to B$ and
linear maps $h_a^*:W\to V$ satisfying
\begin{equation}
h_0=s h_0 r,
\qquad
h_a^*=p_a h_{r(a)}^*q_{h_0(r(a))}.
\label{eq:kar-arrow-components}
\end{equation}
These are precisely the component equations for $h=dhe$.

For $a\in A_0:=\Fix(r)$, $p_a$ is an idempotent. Use its fixed-point
splitting, so
$$
V_a^0=\Fix(p_a)=\im(p_a),
\qquad
V_a^0\xrightarrow{\iota_a^E}V\xrightarrow{\rho_a^E}V_a^0,
$$
with
$$
\rho_a^E\iota_a^E=1,
\qquad
\iota_a^E\rho_a^E=p_a.
$$
Then
$$
\Rlz(E)=(V_a^0)_{a\in A_0}.
$$
If $h:E\to F$, the first equation in
\eqref{eq:kar-arrow-components} restricts $h_0$ to a function
$A_0\to B_0:=\Fix(s)$, and the realized backward map is
\begin{equation}
(\Rlz h)_a^*
=
\rho_a^E h_a^*\iota_{h_0(a)}^F:
W_{h_0(a)}^0\longrightarrow V_a^0.
\label{eq:realized-arrow-components}
\end{equation}

\subsection{Arbitrary small coproducts in the completed model}
\label{app:arbitrary-vect-coproducts}

Let $(E_a)_{a\in A}$ be a small family of objects in $\mathcal K$, with
$$
E_a=\bigl((I_a,V_a),(r_a,p^a)\bigr).
$$
Put
$$
I=\coprod_{a\in A}I_a,
\qquad
P=\bigoplus_{a\in A}V_a,
$$
and write $\jmath_a:V_a\to P$ and $\pi_a:P\to V_a$ for the coordinate maps in $\Vect$. Define $c=(r,p)$ on $(I,P)$ by
\begin{equation}
\begin{aligned}
r(\mathsf{in}_a(i))
&=\mathsf{in}_a(r_a(i)),\\
p_{\mathsf{in}_a(i)}
&=\jmath_a p^a_i\pi_a:P\longrightarrow P.
\end{aligned}
\label{eq:arbitrary-vect-coproduct}
\end{equation}
Since the reverse-target composition law gives
$$
(p\circ p)_{\mathsf{in}_a(i)}
=
p_{\mathsf{in}_a(i)}p_{r(\mathsf{in}_a(i))},
$$
we have
$$
\jmath_a p^a_i\pi_a\jmath_a p^a_{r_a(i)}\pi_a
=
\jmath_a p^a_i p^a_{r_a(i)}\pi_a
=
\jmath_a p^a_i\pi_a.
$$
Thus $c$ is idempotent.

The injection $\iota_a:E_a\to((I,P),c)$ is
\begin{equation}
(\iota_a)_0(i)=\mathsf{in}_a(r_a(i)),
\qquad
(\iota_a)^*_i=p^a_i\pi_a:P\longrightarrow V_a.
\label{eq:arbitrary-vect-injection}
\end{equation}
Let $H=((J,W),(s,q))$ and let $h_a:E_a\to H$ be Karoubi arrows. Define
\begin{equation}
[h_a]_{a\in A,0}(\mathsf{in}_a(i))=h_{a,0}(i),
\qquad
[h_a]_{a\in A,\mathsf{in}_a(i)}^*
=\jmath_a h_{a,i}^*:W\longrightarrow P.
\label{eq:arbitrary-vect-copairing}
\end{equation}
The target compatibility of each $h_a$ and its source compatibility with $(r_a,p^a)$ imply the sandwich equation for \eqref{eq:arbitrary-vect-copairing}. Its composite with \eqref{eq:arbitrary-vect-injection} has primal component
$$
h_{a,0}(r_a(i))=h_{a,0}(i)
$$
and backward component
$$
p^a_i\pi_a\jmath_a h^*_{a,r_a(i)}
=
p^a_i h^*_{a,r_a(i)}
=
h^*_{a,i}.
$$
Hence $[h_a]_{a\in A}\iota_a=h_a$. Conversely, if $k:((I,P),c)\to H$ has these composites, the equation $k=kc$ restricts its backward map on the $a$-th block to the image of $\jmath_a p^a_i\pi_a$, while $k\iota_a=h_a$ fixes that restriction as $\jmath_a h^*_{a,i}$. The direct-sum universal property then forces \eqref{eq:arbitrary-vect-copairing}; the primal component is forced branchwise as well. Thus $((I,P),c)$ is the coproduct of the $E_a$. For $A=\varnothing$, this gives $((\varnothing,0),1)$.

As expected, realization sends this object to the concatenated family
$$
\bigl(\im(p^a_i)\bigr)_{a\in A,\ i\in\Fix(r_a)},
$$
because the image of $\jmath_a p^a_i\pi_a$ is canonically isomorphic to $\im(p^a_i)$. Under these isomorphisms the displayed injections and copairing become the standard family coproduct maps.

\subsection{Initial object and binary coproducts}
The bicartesian closed structure of $\mathcal K$ is the structure used by the
program semantics. The initial object of $\mathcal K$ is
$$
0_{\mathcal K}=((\varnothing,0),1).
$$
There is a unique ambient arrow from it to every object, so it is also
initial in the Karoubi envelope.

Let
$$
E=((A,V),(r,p)),
\qquad
F=((B,W),(s,q)).
$$
Their coproduct is
\begin{equation}
	E+F
	=
	\left(
	(A+B,V\oplus W),c_{E,F}
	\right),
	\label{eq:kar-coproduct-concrete}
\end{equation}
where
$$
\begin{aligned}
	(c_{E,F})_0(\mathsf{inl}(a))&=\mathsf{inl}(r(a)),
	&
	(c_{E,F})^*_{\mathsf{inl}(a)}(v,w)&=(p_a(v),0),\\
	(c_{E,F})_0(\mathsf{inr}(b))&=\mathsf{inr}(s(b)),
	&
	(c_{E,F})^*_{\mathsf{inr}(b)}(v,w)&=(0,q_b(w)).
\end{aligned}
$$
The equations for $e$ and $d$ show branchwise that $c_{E,F}^2=c_{E,F}$.

The injections are
$$
\iota_E:E\longrightarrow E+F,
\qquad
\iota_F:F\longrightarrow E+F,
$$
with
$$
(\iota_E)_0(a)=\mathsf{inl}(r(a)),
\qquad
(\iota_E)_a^*(v,w)=p_a(v),
$$
$$
(\iota_F)_0(b)=\mathsf{inr}(s(b)),
\qquad
(\iota_F)_b^*(v,w)=q_b(w).
$$
These are precisely the ambient injections sandwiched by the source and
coproduct idempotents.

The realization equivalence sends this coproduct to the branchwise family
coproduct because
$$
\Fix((c_{E,F})_0)=\Fix(r)+\Fix(s),
$$
and the projector images on the left and right branches are respectively
$\im(p_a)$ and $\im(q_b)$. The realized injections and copairing are the
standard family ones.

\subsection{Terminal object and binary products}

Let
$$
E=((A,V),(r,p)),
\qquad
F=((B,W),(s,q))
$$
be objects of $\mathcal K$.
The terminal object is
$$
1_{\mathcal K}=((1,0),1).
$$
There is a unique ambient arrow from every object to $(1,0)$, and it
satisfies every sandwich equation because its backward component has domain
$0$.

The product is
\begin{equation}
E\times F
=
\left(
(A\times B,V\oplus W),
(r\times s,p\oplus q)
\right),
\label{eq:kar-product}
\end{equation}
where
$$
(r\times s)(a,b)=(r(a),s(b)),
$$
$$
(p\oplus q)_{(a,b)}(v,w)=(p_a(v),q_b(w)).
$$
Its projections are
$$
\pi_E:E\times F\longrightarrow E,
\qquad
\pi_F:E\times F\longrightarrow F,
$$
with primal components
$$
(\pi_E)_0(a,b)=r(a),
\qquad
(\pi_F)_0(a,b)=s(b),
$$
and backward components
$$
(\pi_E)^*_{(a,b)}(v)=(p_a(v),0),
\qquad
(\pi_F)^*_{(a,b)}(w)=(0,q_b(w)).
$$
They satisfy the sandwich equations directly.

Let
$$
G=((C,U),(t,\ell))
$$
and let $h:G\to E$ and $k:G\to F$. Define
$$
\langle h,k\rangle_0(c)=(h_0(c),k_0(c))
$$
and
\begin{equation}
\langle h,k\rangle_c^*(v,w)
=h_c^*(v)+k_c^*(w).
\label{eq:kar-product-pairing}
\end{equation}
The sandwich equations for $h$ and $k$ imply
$$
\langle h,k\rangle
=(e\times d)\langle h,k\rangle(t,\ell),
$$
so this is a Karoubi arrow. Moreover,
$$
\pi_E\langle h,k\rangle=h,
\qquad
\pi_F\langle h,k\rangle=k,
$$
because target compatibility gives $h_c^*p_{h_0(c)}=h_c^*$ and
$k_c^*q_{k_0(c)}=k_c^*$.

For uniqueness, suppose $m:G\to E\times F$ has these two composites. The
projection equations force
$$
m_0(c)=(h_0(c),k_0(c)).
$$
Writing $m_c^*:V\oplus W\to U$, the same equations determine its restrictions
to the two biproduct summands:
$$
m_c^*(v,0)=h_c^*(v),
\qquad
m_c^*(0,w)=k_c^*(w).
$$
Linearity therefore forces \eqref{eq:kar-product-pairing}. This proves the
product universal property.

As expected, realization preserves this product through the canonical isomorphisms
$$
\Fix(r\times s)=\Fix(r)\times\Fix(s),
$$
$$
\im(p_a\oplus q_b)
\cong
\im(p_a)\oplus\im(q_b).
$$
Under these isomorphisms the realized projections and pairing are the
standard family projections and pairing.

\subsection{The standard dependent-family structure}

We first record the structure to which the ambient formulas will be compared.
We recall the structure of $\mathcal F$, see~\cite[Sections~9.2 and~10.2]{LucatelliNunesVakar2023}.
Let
$$
X=(V_a)_{a\in A},
\qquad
Y=(W_b)_{b\in B}
$$
be objects of $\mathcal F$.
The terminal object and binary product are
$$
1_{\mathcal F}=(0)_{*\in1},
\qquad
X\times Y=(V_a\oplus W_b)_{(a,b)\in A\times B}.
$$
The initial object is the empty family, and the binary coproduct is
branchwise:
$$
0_{\mathcal F}=()_{c\in\varnothing},
\qquad
X+Y=(Z_c)_{c\in A+B},
$$
where
$$
Z_{\mathsf{inl}(a)}=V_a,
\qquad
Z_{\mathsf{inr}(b)}=W_b.
$$

For the exponential, define
$$
\mathsf H(X,Y)=
\left\{
(f,\phi)\ \middle|\
\begin{array}{l}
	f:A\to B,\\[-1mm]
	\phi_a:W_{f(a)}\to V_a\text{ linear for all }a\in A
\end{array}
\right\}.
$$
For $(f,\phi)\in\mathsf H(X,Y)$, put
$$
K_{(f,\phi)}=\bigoplus_{a\in A}W_{f(a)}.
$$
Then
\begin{equation}
	Y^X=
	\bigl(K_{(f,\phi)}\bigr)_{(f,\phi)\in\mathsf H(X,Y)}.
	\label{eq:fam-exponential}
\end{equation}
The evaluation arrow has primal component
$$
\mathsf{ev}_0((f,\phi),a)=f(a)
$$
and backward component
\begin{equation}
	\mathsf{ev}_{((f,\phi),a)}^*(w)
	=
	\bigl(\jmath_a(w),\phi_a(w)\bigr):
	W_{f(a)}\longrightarrow K_{(f,\phi)}\oplus V_a,
	\label{eq:fam-evaluation}
\end{equation}
where $\jmath_a$ is the $a$-th direct-sum injection.

For completeness, let $Z=(U_c)_{c\in C}$ and let
$k:Z\times X\to Y$. Write its backward component as
$$
k_{c,a}^*=\langle k^U_{c,a},k^V_{c,a}\rangle:
W_{k_0(c,a)}\longrightarrow U_c\oplus V_a.
$$
Its transpose sends $c$ to
$$
\left(
a\longmapsto k_0(c,a),
(k^V_{c,a})_{a\in A}
\right)\in\mathsf H(X,Y),
$$
and its backward component
$$
\bigoplus_{a\in A}W_{k_0(c,a)}\longrightarrow U_c
$$
has restriction $k^U_{c,a}$ to the $a$-th summand. Evaluation recovers
$k_0$, $k^U$, and $k^V$. Hence these operations are mutually inverse and
natural, proving the exponential adjunction in $\mathcal F$.

\subsection{The ambient exponential and its Karoubi idempotent}

The ambient category $\Copow(\Vectop)$ is cartesian closed; the following
formulas are the concrete constant-cotangent instance of basic higher-order
CHAD~\cite{Vakar2021,VakarSmeding2022}. Its exponential from $(A,V)$ to
$(B,W)$ has primal set
$$
H=
\left\{
(f,\phi)\ \middle|\
f:A\to B,
\ \phi_a:W\to V\text{ linear for all }a\in A
\right\}
$$
and cotangent space
$$
Z=\bigoplus_{a\in A}W.
$$
Let $\jmath_a:W\to Z$ be the direct-sum injections. Ambient evaluation has
primal component
$$
\mathsf{ev}_0((f,\phi),a)=f(a)
$$
and backward component
\begin{equation}
\mathsf{ev}_{((f,\phi),a)}^*(w)
=
(\jmath_a(w),\phi_a(w)):
W\longrightarrow Z\oplus V.
\label{eq:ambient-evaluation}
\end{equation}
The ordinary currying operation is obtained by decomposing a backward map
into its $Z$- and $V$-components and using the direct-sum universal property.

The standard fact that the Karoubi envelope of a cartesian-closed category is
cartesian closed belongs to the classical Karoubi/semifunctor theory~\cite{Hoofman1993,HoofmanMoerdijk1995}; see also the general Cauchy-completion
background in~\cite{BorceuxDejean1986}. Here we need its concrete form.
Define
$$
\varepsilon_{E,F}:(H,Z)\longrightarrow(H,Z)
$$
by
\begin{equation}
(\varepsilon_{E,F})_0(f,\phi)
=
\left(
 s f r,
 \bigl(p_a\phi_{r(a)}q_{f(r(a))}\bigr)_{a\in A}
\right)
\label{eq:exponential-idempotent-base}
\end{equation}
and by the unique backward map satisfying
\begin{equation}
(\varepsilon_{E,F})^*_{(f,\phi)}\jmath_a
=
\jmath_{r(a)}q_{f(r(a))}
\qquad(a\in A).
\label{eq:exponential-idempotent-back}
\end{equation}
This is the transpose of
$$
(H,Z)\times(A,V)
\xrightarrow{1\times e}
(H,Z)\times(A,V)
\xrightarrow{\mathsf{ev}}
(B,W)
\xrightarrow{d}
(B,W).
$$
Equivalently, it is the internal operation
$$
h\longmapsto dhe.
$$
It is therefore idempotent. Directly, the base calculation uses
$r^2=r$, $s^2=s$, $p_ap_{r(a)}=p_a$, and $q_bq_{s(b)}=q_b$; the backward
calculation follows by checking every direct-sum injection in
\eqref{eq:exponential-idempotent-back}.

The exponential in $\mathcal K$ is
\begin{equation}
F^E=((H,Z),\varepsilon_{E,F}).
\label{eq:kar-exponential}
\end{equation}
Its evaluation is the restricted ambient evaluation
$$
\mathsf{ev}_{E,F}
=
d\,\mathsf{ev}\,(\varepsilon_{E,F}\times e).
$$
At an ambient point $((f,\phi),a)$, its primal component is
$$
s(f(r(a)))
$$
and its backward component is
\begin{equation}
w\longmapsto
\left(
 \jmath_{r(a)}q_{f(r(a))}(w),
 p_a\phi_{r(a)}q_{f(r(a))}(w)
\right).
\label{eq:kar-evaluation-explicit}
\end{equation}

Let
$$
G=((C,U),t)
$$
and let $k:G\times E\to F$ be a Karoubi arrow. Write its ambient backward
map as
$$
k_{c,a}^*=\langle k^U_{c,a},k^V_{c,a}\rangle.
$$
Its ambient transpose has primal component
$$
c\longmapsto
\left(
 a\longmapsto k_0(c,a),
 (k^V_{c,a})_{a\in A}
\right)
$$
and backward component $Z\to U$ whose restriction to the $a$-th summand is
$k^U_{c,a}$. The ambient exponential adjunction transports the sandwich
equation exactly as follows:
$$
k=d\,k\,(t\times e)
\quad\Longleftrightarrow\quad
\Lambda(k)=\varepsilon_{E,F}\,\Lambda(k)\,t.
$$
Hence $\Lambda(k)$ is a Karoubi arrow $G\to F^E$. Conversely, restricted
evaluation uncarries every Karoubi arrow $G\to F^E$. Since these are the
ordinary mutually inverse currying operations, they are natural and satisfy
the triangle equations. This proves the exponential universal property.

\subsection{Explicit comparison of the exponentials}

Put
$$
A_0=\Fix(r),
\qquad
B_0=\Fix(s),
\qquad
V_a^0=\im(p_a),
\qquad
W_b^0=\im(q_b).
$$
The base set of $\Rlz(F)^{\Rlz(E)}$ consists of pairs
$$
(g,\bar\phi),
\qquad
g:A_0\to B_0,
\qquad
\bar\phi_a:W_{g(a)}^0\to V_a^0.
$$
There is a bijection
\begin{equation}
\Theta:
\Fix((\varepsilon_{E,F})_0)
\longrightarrow
\mathsf H(\Rlz(E),\Rlz(F))
\label{eq:exponential-base-comparison}
\end{equation}
given by
$$
\Theta(f,\phi)
=
\left(
 f|_{A_0},
 (\rho_a^E\phi_a\iota_{f(a)}^F)_{a\in A_0}
\right).
$$
The inverse bijection sends $(g,\bar\phi)$ to
$(\widetilde g,\widetilde\phi)$, where
$$
\widetilde g(a)=g(r(a)),
$$
$$
\widetilde\phi_a
=
p_a\iota_{r(a)}^E
\bar\phi_{r(a)}
\rho_{g(r(a))}^F.
$$
The fixed-point equations in
\eqref{eq:exponential-idempotent-base} show directly that the two assignments
are inverse.

Fix $(f,\phi)\in\Fix((\varepsilon_{E,F})_0)$ and put
$$
K_{(f,\phi)}^0
=
\bigoplus_{a\in A_0}W_{f(a)}^0.
$$
This is the fibre of the family exponential at $\Theta(f,\phi)$. Define
$$
K_{(f,\phi)}^0
\xrightarrow{\iota_{(f,\phi)}^{\Rightarrow}}
Z
\xrightarrow{\rho_{(f,\phi)}^{\Rightarrow}}
K_{(f,\phi)}^0
$$
by
$$
\iota_{(f,\phi)}^{\Rightarrow}\jmath_a^0
=
\jmath_a\iota_{f(a)}^F
\qquad(a\in A_0),
$$
$$
\rho_{(f,\phi)}^{\Rightarrow}\jmath_a
=
\jmath_{r(a)}^0\rho_{f(r(a))}^F
\qquad(a\in A).
$$
Then
$$
\rho_{(f,\phi)}^{\Rightarrow}
\iota_{(f,\phi)}^{\Rightarrow}=1,
$$
and
$$
\iota_{(f,\phi)}^{\Rightarrow}
\rho_{(f,\phi)}^{\Rightarrow}
=(\varepsilon_{E,F})^*_{(f,\phi)}.
$$
The second equation follows by evaluating both maps on each $\jmath_a$ and
using \eqref{eq:exponential-idempotent-back}. Thus realization of the
ambient exponential fibre gives precisely the family exponential fibre.

Under the base bijection and these splittings, realization of
\eqref{eq:kar-evaluation-explicit} is the family evaluation
\eqref{eq:fam-evaluation}. Indeed, for fixed $(f,\phi)$, fixed $a\in A_0$,
and $w\in W_{f(a)}^0$, the realized backward map is
$$
w\longmapsto
\left(
 \jmath_a^0(w),
 \rho_a^E\phi_a\iota_{f(a)}^F(w)
\right).
$$
The first component is the cotangent contribution to the function object and
the second is the realized backward map at $a$. Since realization identifies
evaluation, it also identifies currying by uniqueness of transposes.

\section{Finite coproducts in the completed reverse target}
\label{app:target-closure}

We now pass from the constant-family theorem to the ordinary reverse target of
basic CHAD\@. The cartesian-closed part of the argument is already available:
basic CHAD proves that the total reverse category is cartesian closed, and
Karoubi completion preserves cartesian closure. The only new task is therefore
to construct the finite coproducts required by variants.

\subsection{Inherited cartesian closure}

Let
$$
\TotTar=\Sigma_{\CTar}\LTar^{\mathrm{op}}
$$
be the ordinary reverse target. We assume that $(\CTar,\LTar)$ is a
categorical model of the basic CHAD target language in the sense of
\cite[Definition~5.5]{VakarSmeding2022}. In particular, $\LTar$ is locally
indexed, its fibres have chosen finite biproducts preserved by reindexing, and
it supports cartesian linear-function objects. By
\cite[Theorem~6.2]{VakarSmeding2022}, $\TotTar$ is cartesian closed. For
variants, we assume only that $\CTar$ has finite coproducts.

\begin{proposition}[cartesian closure under Karoubi completion]
\label{prop:kar-ccc-app}
If $\mathcal A$ is cartesian closed, then $\Kar(\mathcal A)$ is
cartesian closed. For Karoubi objects $(A,e)$ and $(B,d)$, the exponential may
be taken to be
$$
(B,d)^{(A,e)}=(B^A,\varepsilon_{e,d}),
$$
where
$$
\varepsilon_{e,d}
=
\Lambda\bigl(d\circ\mathrm{ev}\circ(1_{B^A}\times e)\bigr),
\qquad
\varepsilon_{e,d}(h)=dhe.
$$
The terminal object and products are obtained by applying the corresponding
idempotents componentwise.
\end{proposition}

This is the standard cartesian-closed structure on the Karoubi completion; see
\cite{Hoofman1993,HoofmanMoerdijk1995}. Applying it to $\TotTar$ shows that
$\Kar(\TotTar)$ is already cartesian closed. We now construct its finite
coproducts from the structure already present in the basic CHAD target.

\subsection{Common hosts and branchwise linear maps}

\begin{lemma}[Common hosts in the reverse fibres]
\label{lem:target-common-hosts}
Let $X\in\CTar$, and let $(V_i)_{i\in I}$ be a finite family of linear
objects. Then $(V_i)_{i\in I}$ admits a common retract host in
$\LTar(X)^{\mathrm{op}}$. One may take
$$
P=\bigoplus_{i\in I}V_i.
$$
These common-host presentations are preserved by reindexing.
\end{lemma}

\begin{proof}
In $\LTar(X)$, let $\iota_i:V_i\to P$ and $\pi_i:P\to V_i$ be the
chosen biproduct injection and projection. Since
$\pi_i\iota_i=1_{V_i}$, the reversed arrows
$$
V_i\xrightarrow{\pi_i^{\mathrm{op}}}P
\xrightarrow{\iota_i^{\mathrm{op}}}V_i
$$
form a retract presentation in $\LTar(X)^{\mathrm{op}}$. Local indexedness
makes the linear objects independent of the base, and preservation of the
chosen biproducts by reindexing makes these presentations stable under
substitution.
\end{proof}

\begin{lemma}[Branchwise decomposition of linear maps]
\label{lem:target-linear-case}
For $A,B\in\CTar$ and linear objects $U,P$, restriction along the coproduct
injections induces a natural bijection
$$
\LTar(A+B)(U,P)
\cong
\LTar(A)(U,P)\times\LTar(B)(U,P).
$$
Consequently, a linear map over $A+B$ is uniquely determined by its two branch
restrictions. The induced branchwise case operation is stable under reindexing
and satisfies the expected $\beta$- and $\eta$-equations.
\end{lemma}

\begin{proof}
The cartesian linear-function object $U\multimap P$ represents linear maps
in the base variable, giving natural bijections
$$
\LTar(X)(U,P)\cong\CTar(X,U\multimap P).
$$
Hence
$$
\begin{aligned}
\LTar(A+B)(U,P)
&\cong\CTar(A+B,U\multimap P)\\
&\cong\CTar(A,U\multimap P)
\times\CTar(B,U\multimap P)\\
&\cong\LTar(A)(U,P)\times\LTar(B)(U,P).
\end{aligned}
$$
Naturality identifies this bijection with restriction along
$\mathsf{inl}$ and $\mathsf{inr}$. Its inverse is branchwise linear case
analysis, and the two inverse equations give the $\beta$- and
$\eta$-equations.
\end{proof}

Thus the two pieces of structure that might otherwise appear as additional
variant assumptions are already available: the fibrewise biproducts give the
common hosts, while the cartesian linear-function objects derive branchwise
linear case analysis from finite coproducts in $\CTar$.

\subsection{Finite coproducts after completion}

Let
$$
E=((A,V),e),
\qquad
F=((B,W),d),
$$
where $e=(r,p)$ and $d=(s,q)$ are idempotent endomorphisms in $\TotTar$.
Put $P=V\oplus W$. Define an endomorphism $c_{E,F}$ of
$(A+B,P)$ by
$$
\begin{aligned}
c_{E,F,1}(\mathsf{inl}(a))
&=\mathsf{inl}(r(a)),
&
(c_{E,F})^*_{\mathsf{inl}(a)}(v,w)
&=(p_a(v),0),\\
c_{E,F,1}(\mathsf{inr}(b))
&=\mathsf{inr}(s(b)),
&
(c_{E,F})^*_{\mathsf{inr}(b)}(v,w)
&=(0,q_b(w)).
\end{aligned}
$$
The backward component exists uniquely by
\Cref{lem:target-linear-case}. Since
$$
r^2=r,
\qquad
p_a\circ p_{r(a)}=p_a,
\qquad
s^2=s,
\qquad
q_b\circ q_{s(b)}=q_b,
$$
the same branchwise calculation gives
$$
c_{E,F}^2=c_{E,F}.
$$
Set
$$
E+F:=((A+B,V\oplus W),c_{E,F}).
$$

Let
$$
j_E:(A,V)\longrightarrow(A+B,V\oplus W),
\qquad
j_F:(B,W)\longrightarrow(A+B,V\oplus W)
$$
be the ambient injections. Their primal components are
$\mathsf{inl}$ and $\mathsf{inr}$, and their backward components are the
corresponding biproduct projections. Define the completed injections by
$$
\iota_E=c_{E,F}\circ j_E\circ e,
\qquad
\iota_F=c_{E,F}\circ j_F\circ d.
$$
Explicitly,
$$
(\iota_E)_1=\mathsf{inl}\circ r,
\qquad
(\iota_E)^*_a(v,w)=p_a(v),
$$
and
$$
(\iota_F)_1=\mathsf{inr}\circ s,
\qquad
(\iota_F)^*_b(v,w)=q_b(w).
$$

\begin{proposition}[Finite coproducts in the completed reverse target]
\label{prop:target-coproducts-app}
If $\CTar$ has finite coproducts, then $\Kar(\TotTar)$ has finite
coproducts. The binary coproduct of $E$ and $F$ is the object $E+F$ above, and
the initial object is
$$
((0,0),1_{(0,0)}),
$$
where the first $0$ is initial in $\CTar$ and the second is the zero linear
object.
\end{proposition}

\begin{proof}
Let $H=((C,U),h)$, and let $a:E\to H$ and $b:F\to H$ be Karoubi arrows.
Define an ambient arrow
$$
[a,b]:(A+B,V\oplus W)\longrightarrow(C,U)
$$
by ordinary primal copairing and by the branch restrictions
$$
[a,b]^*_{\mathsf{inl}(x)}(u)=(a_x^*(u),0),
\qquad
[a,b]^*_{\mathsf{inr}(y)}(u)=(0,b_y^*(u)).
$$
The latter determine a unique linear map by
\Cref{lem:target-linear-case}.

The equations $a=h\circ a\circ e$ and $b=h\circ b\circ d$ imply,
branchwise,
$$
[a,b]=h\circ[a,b]\circ c_{E,F},
$$
so $[a,b]$ is a Karoubi arrow $E+F\to H$. The source-compatibility equations
$a=a\circ e$ and $b=b\circ d$ give
$$
[a,b]\circ\iota_E=a,
\qquad
[a,b]\circ\iota_F=b.
$$
For example, the left backward component reduces to
$$
p_x\circ a^*_{r(x)}=a_x^*.
$$

Conversely, let $k:E+F\to H$ satisfy
$k\circ\iota_E=a$ and $k\circ\iota_F=b$. Since
$k=k\circ c_{E,F}$, on the left branch one has
$$
k_1(\mathsf{inl}(x))
=
k_1(\mathsf{inl}(r(x)))
$$
and
$$
k^*_{\mathsf{inl}(x)}(u)
=
\bigl(p_x\pi_Vk^*_{\mathsf{inl}(r(x))}(u),0\bigr),
$$
where $\pi_V:V\oplus W\to V$ is the first biproduct projection. The
equation $k\circ\iota_E=a$ then forces
$$
k_1(\mathsf{inl}(x))=a_1(x),
\qquad
k^*_{\mathsf{inl}(x)}(u)=(a_x^*(u),0).
$$
The right branch is symmetric. The primal coproduct $\eta$-equation and the
linear $\eta$-equation of \Cref{lem:target-linear-case} yield
$k=[a,b]$. This proves the binary coproduct universal property.

Finally, from $(0,0)$ to any ambient object there is a unique arrow: its primal
component is the unique map from the initial object, and its backward component
is the unique linear map into the zero object. Composing this arrow with any
target idempotent leaves it unchanged, so it is also the unique arrow from
$((0,0),1_{(0,0)})$ in the Karoubi completion.
\end{proof}

\begin{corollary}[Bicartesian closure of the completed reverse target]
\label{cor:target-bcc-app}
Let $(\CTar,\LTar)$ be a categorical model of the basic CHAD target language.
If $\CTar$ has finite coproducts, then
$$
\Kar\left(\Sigma_{\CTar}\LTar^{\mathrm{op}}\right)
$$
is bicartesian closed.
\end{corollary}

\begin{proof}
Cartesian closure follows from \cite[Theorem~6.2]{VakarSmeding2022} and
\Cref{prop:kar-ccc-app}; finite coproducts follow from
\Cref{prop:target-coproducts-app}.
\end{proof}

\section{The simply typed reverse-mode transformation}
\label{app:macro}

Bicartesian closure determines the transformation abstractly, following the semantics-driven methodology of basic and expressive CHAD~\cite{Vakar2021,VakarSmeding2022,LucatelliNunesVakar2023}. This section unfolds the new data-type clauses into executable nondependent target terms and records the categorical form of the higher-order clauses. The displayed term rules below are for data-type contexts and results, where the primal component of every generated idempotent is the identity. For arbitrary higher-order contexts, structural maps must instead be restricted by the full source and target idempotents. The invariant throughout is the sandwich equation: each generated map restricts to the valid projector image at its source and lands in the valid image at its target. The presentation also separates this semantic derivation from subsequent implementation optimisations such as sparse accumulation and closure conversion~\cite{SmedingVakar2024}.

\subsection{Types, ambient cotangents, and projectors}

For clarity, we spell out the clauses that are new for variants or modified by the generated projectors. The data types of the present source fragment are
$$
\tau,\sigma::=
\mathsf{real}^n
\mid 1
\mid \tau\times\sigma
\mid 0
\mid \tau+\sigma.
$$
For each type, the translation produces a primal type $\widehat D_1\tau$, an ambient cotangent type $\widehat D_2\tau$, and a primal-indexed projector
$$
p_\tau:
\widehat D_1\tau;
\widehat D_2\tau
\multimap
\widehat D_2\tau.
$$
At ground types and the unit type,
$$
\widehat D_1(\mathsf{real}^n)=\mathsf{real}^n,
\qquad
\widehat D_2(\mathsf{real}^n)=\mathsf{real}^n,
\qquad
p_{\mathsf{real}^n}(x,v)=v,
$$
$$
\widehat D_1(1)=1,
\qquad
\widehat D_2(1)=0,
\qquad
p_1(*,0)=0.
$$
For products,
$$
\begin{aligned}
\widehat D_1(\tau\times\sigma)
&=\widehat D_1\tau\times\widehat D_1\sigma,\\
\widehat D_2(\tau\times\sigma)
&=\widehat D_2\tau\oplus\widehat D_2\sigma,\\
p_{\tau\times\sigma}((x,y),(v,w))
&=(p_\tau(x,v),p_\sigma(y,w)).
\end{aligned}
$$
For sums,
$$
\begin{aligned}
\widehat D_1(\tau+\sigma)
&=\widehat D_1\tau+\widehat D_1\sigma,\\
\widehat D_2(\tau+\sigma)
&=\widehat D_2\tau\oplus\widehat D_2\sigma,\\
p_{\tau+\sigma}(\mathsf{inl}(x),(v,w))
&=(p_\tau(x,v),0),\\
p_{\tau+\sigma}(\mathsf{inr}(y),(v,w))
&=(0,p_\sigma(y,w)).
\end{aligned}
$$
The empty type is translated by $\widehat D_1(0)=0$ and
$\widehat D_2(0)=0$; it has no primal values, so its projector has no branch to
define. For a context
$$
\Gamma=x_1:\tau_1,\ldots,x_n:\tau_n,
$$
put $\widehat D_2\Gamma=\bigoplus_i\widehat D_2\tau_i$ and let
$p_\Gamma(\gamma,-)$ act componentwise at the primal environment $\gamma$.
We write $\onehot_i:\widehat D_2\tau_i\multimap\widehat D_2\Gamma$ for the
$i$-th biproduct injection.

\subsection{Clauses for data-type constructors}

For a judgement $\Gamma\vdash t:\tau$ in the data-type fragment, the transformation produces
$$
\widehat D_1\Gamma\vdash
\overleftarrow D_\Gamma(t):
\widehat D_1\tau
\times
(\widehat D_2\tau\multimap\widehat D_2\Gamma).
$$
Write $0$ and $+$ for cotangent zero and addition. The standard CHAD clauses are sandwiched by the generated projectors. We display the projector-modified structural clauses and then give the variant clauses in full; primitive clauses not affected by variants are inherited from basic CHAD~\cite[Section~7]{VakarSmeding2022}. The displayed variable clause is valid because all types in $\Gamma$ are data types. In an arbitrary higher-order context, let $e_\Gamma$ denote the product idempotent on the translated context. The structural arrow is the full Karoubi restriction
$$
\Dst(x_i)=e_{\tau_i}\,\pi_i\,e_\Gamma.
$$
Thus, for a one-variable context, the erased translation of the identity is
$e_\tau$, not the ambient identity. Its data-type specialization is
$$
\overleftarrow D_\Gamma(x_i)
=
\left\langle
x_i,
\lambda\bar x.\onehot_i(p_{\tau_i}(x_i,\bar x))
\right\rangle,
$$
$$
\overleftarrow D_\Gamma(\langle\rangle)
=
\langle\langle\rangle,\lambda\bar u.0\rangle,
$$
$$
\begin{aligned}
\overleftarrow D_\Gamma(\langle t,s\rangle)
={}&
\mathbf{let}\ \langle x,x^*\rangle=\overleftarrow D_\Gamma(t)\ \mathbf{in}\\
&\mathbf{let}\ \langle y,y^*\rangle=\overleftarrow D_\Gamma(s)\ \mathbf{in}\\
&\left\langle
\langle x,y\rangle,
\lambda(\bar x,\bar y).
 x^*(p_\tau(x,\bar x))+y^*(p_\sigma(y,\bar y))
\right\rangle,
\end{aligned}
$$
$$
\begin{aligned}
\overleftarrow D_\Gamma(\pi_1t)
={}&
\mathbf{let}\ \langle(x,y),h\rangle=\overleftarrow D_\Gamma(t)\ \mathbf{in}\\
&\left\langle x,\lambda\bar x.h(p_\tau(x,\bar x),0)\right\rangle,
\end{aligned}
$$
with a symmetric clause for $\pi_2$.

For a primitive operation, use its chosen codifferential and compose on both sides with the generated projectors. For a let-binding, suppose
$$
\overleftarrow D_\Gamma(t)=\langle x,x^*\rangle,
\qquad
\overleftarrow D_{\Gamma,x:\tau}(s)=\langle y,y^*\rangle.
$$
Then
$$
\begin{aligned}
&\overleftarrow D_\Gamma(\mathbf{let}\ x=t\ \mathbf{in}\ s)\\
={}&
\mathbf{let}\ \langle x,x^*\rangle=\overleftarrow D_\Gamma(t)\ \mathbf{in}\\
&\mathbf{let}\ \langle y,y^*\rangle=\overleftarrow D_{\Gamma,x:\tau}(s)\ \mathbf{in}\\
&\left\langle
 y,
 \lambda\bar y.
 \mathbf{let}\ \langle\bar\gamma,\bar x\rangle
 =y^*(p_\sigma(y,\bar y))\ \mathbf{in}
 p_\Gamma(\gamma,\bar\gamma)+x^*(p_\tau(x,\bar x))
\right\rangle.
\end{aligned}
$$

For injections,
$$
\begin{aligned}
\overleftarrow D_\Gamma(\mathsf{inl}\,t)
={}&
\mathbf{let}\ \langle x,x^*\rangle=\overleftarrow D_\Gamma(t)\ \mathbf{in}\\
&\left\langle
\mathsf{inl}(x),
\lambda(v,w).x^*(p_\tau(x,v))
\right\rangle,
\end{aligned}
$$
$$
\begin{aligned}
\overleftarrow D_\Gamma(\mathsf{inr}\,t)
={}&
\mathbf{let}\ \langle y,y^*\rangle=\overleftarrow D_\Gamma(t)\ \mathbf{in}\\
&\left\langle
\mathsf{inr}(y),
\lambda(v,w).y^*(p_\sigma(y,w))
\right\rangle.
\end{aligned}
$$
The ignored ambient coordinate is discarded.

For
$$
c=
\mathsf{case}\ t\ \mathsf{of}\
\{\mathsf{inl}(x)\mapsto s
\mid
\mathsf{inr}(y)\mapsto r\},
$$
the complete clause is
$$
\begin{aligned}
\overleftarrow D_\Gamma(c)
={}&\mathbf{let}\ \langle z,z^*\rangle=
\overleftarrow D_\Gamma(t)\ \mathbf{in}\\
&\mathsf{case}\ z\ \mathsf{of}\\
&\quad\mathsf{inl}(x)\mapsto
\mathbf{let}\ \langle u,u^*\rangle=
\overleftarrow D_{\Gamma,x:\tau}(s)\ \mathbf{in}\\
&\qquad\left\langle u,\lambda\bar u.
\mathbf{let}\ \langle\bar\gamma,\bar x\rangle=
 u^*(p_\rho(u,\bar u))\ \mathbf{in}\right.\\[-.2em]
&\hspace{9.7em}\left.
 p_\Gamma(\gamma,\bar\gamma)+
 z^*\bigl(p_{\tau+\sigma}(\mathsf{inl}(x),(p_\tau(x,\bar x),0))\bigr)
\right\rangle,\\
&\quad\mathsf{inr}(y)\mapsto
\mathbf{let}\ \langle v,v^*\rangle=
\overleftarrow D_{\Gamma,y:\sigma}(r)\ \mathbf{in}\\
&\qquad\left\langle v,\lambda\bar v.
\mathbf{let}\ \langle\bar\gamma,\bar y\rangle=
 v^*(p_\rho(v,\bar v))\ \mathbf{in}\right.\\[-.2em]
&\hspace{9.7em}\left.
 p_\Gamma(\gamma,\bar\gamma)+
 z^*\bigl(p_{\tau+\sigma}(\mathsf{inr}(y),(0,p_\sigma(y,\bar y)))\bigr)
\right\rangle.
\end{aligned}
$$
The scrutinee is evaluated once, only the selected branch is differentiated, and the branch cotangent is inserted into the ambient biproduct before the scrutinee backpropagator is called. Idempotence simplifies the nested projector to the expected branch insertion. The abort term is the unique map induced by the initial completed object.

\subsection{Higher-order types}

Function types use the cartesian closed structure of $\Kar(\TotTar)$. If the domain and codomain are represented by full idempotents $e_\tau$ and $e_\sigma$, then
$$
e_{\tau\Rightarrow\sigma}=\varepsilon_{e_\tau,e_\sigma},
$$
where the exponential idempotent acts internally by
$$
h\longmapsto e_\sigma h e_\tau.
$$
The ambient primal and cotangent types are those of basic higher-order CHAD, but the primal component of $e_{\tau\Rightarrow\sigma}$ generally projects an ambient function--backpropagator package and is not the identity. The concrete base and backward components in the constant-cotangent model are given by \eqref{eq:exponential-idempotent-base} and \eqref{eq:exponential-idempotent-back}. For arbitrary contexts, every structural arrow is similarly sandwiched by the full context and result idempotents. In particular, abstraction is restricted currying and application is restricted evaluation:
$$
\operatorname{ev}_{e_\tau,e_\sigma}
=
e_\sigma\,\operatorname{ev}(\varepsilon_{e_\tau,e_\sigma}\times e_\tau),
\qquad
\Lambda_{e_\tau,e_\sigma}(f)=\varepsilon_{e_\tau,e_\sigma}\Lambda(f).
$$
These are the operations described in \Cref{app:target-closure}. Hence the higher-order clauses do not reuse the data-type variable rule: they are forced by cartesian closedness and use the full idempotents, as required by \Cref{prop:semi}.

\section{Correctness by comparison with dependent CHAD}
\label{app:correctness}

This appendix supplies the concrete comparison used in
\Cref{prop:standard-target-compatibility,thm:comparison} and expands the proof
of \Cref{thm:ambient}.

\subsection{Concrete realization and preservation}
\label{app:concrete-comparison}

Recall the concrete-model assumption of \Cref{sec:correctness}. The primal
functor $\mathcal S:\CTar\to\Set$ and indexed linear functors
$\mathcal L_A:\LTar(A)\to\Vect^{\mathcal S A}$ induce
$$
\llbracket-\rrbracket_\Tar:\TotTar\longrightarrow\Copow(\Vectop),
$$
and therefore
$$
\llbracket-\rrbracket_{\mathrm{cpl}}:
\Kar(\TotTar)
\xrightarrow{\Kar(\llbracket-\rrbracket_\Tar)}
\Kar(\Copow(\Vectop))
\xrightarrow{\Rlz}
\Fam(\Vectop).
$$
The terminal object, products, and exponentials of $\Kar(\TotTar)$ are the
standard Karoubi restrictions of the corresponding structure of $\TotTar$.
They are preserved by the assumed basic-CHAD comparisons and by the explicit
realization comparisons of \Cref{thm:explicit-bcc-comparison}. It remains to
check the completed coproduct.

Let
$$
E=((A,V),(r,p)),
\qquad
F=((B,W),(s,q)).
$$
In the following calculation, $r,s$ denote the functions obtained by applying
$\mathcal S$, while $p_a,q_b$ denote the linear maps obtained by applying the
corresponding indexed interpretations. The coproduct idempotent $c_{E,F}$ has
fixed primal points
$$
\mathsf{inl}(a)\quad(a\in\Fix(r)),
\qquad
\mathsf{inr}(b)\quad(b\in\Fix(s)).
$$
At a left fixed point its backward component on
$\mathcal V V\oplus\mathcal V W$ is
$(v,w)\mapsto(p_a(v),0)$, whose image is canonically isomorphic to
$\im(p_a)$. At a right fixed point it is
$(v,w)\mapsto(0,q_b(w))$, whose image is canonically isomorphic to
$\im(q_b)$. Consequently,
$$
\llbracket E+F\rrbracket_{\mathrm{cpl}}
\cong
\llbracket E\rrbracket_{\mathrm{cpl}}
+
\llbracket F\rrbracket_{\mathrm{cpl}}.
$$
Under these image isomorphisms the completed injections and copairing become
the family injections and branchwise copairing; the empty coproduct becomes the
empty family. This proves \Cref{prop:standard-target-compatibility}.

Let $\Ddep:\Synp\to\Fam(\Vectop)$ be dependent CHAD. Both
$\llbracket-\rrbracket_{\mathrm{cpl}}\DKar$ and $\Ddep$ are bicartesian closed
interpretations, and by assumption they agree on every ground type and
primitive operation. They satisfy the same declared equations, so the universal
property of $\Synp$ gives the coherent natural isomorphism of
\Cref{thm:comparison}.

Every data type in the present grammar is isomorphic, by repeated distributivity,
to a finite sum of finite products of Euclidean ground types. The dependent
correctness theorem applies to that normal form, and naturality with respect to
the structural isomorphisms transports the result back to the original type.
Thus the realized backward map is the mathematical transposed derivative on the
genuine cotangent fibres.

\subsection{The exact ambient formula}

Let $t:\tau\to\sigma$ have data-type domain and codomain, write
$f=\llbracket t\rrbracket$, and use the fixed-point splitting maps determined
by the comparison with dependent CHAD
$$
T^*_\tau(x)
\xrightarrow{\iota_{\tau,x}}
\widehat T^*_\tau
\xrightarrow{\rho_{\tau,x}}
T^*_\tau(x),
\qquad
\rho_{\tau,x}\iota_{\tau,x}=1,
\quad
\iota_{\tau,x}\rho_{\tau,x}=p_{\tau,x}.
$$
Let $h=\widehat R_t(x)$ be the ambient backward map. The comparison with
dependent CHAD gives
$$
\rho_{\tau,x}h\iota_{\sigma,f(x)}
=
Df(x)^{\mathsf T},
$$
while the Karoubi equation gives
$$
h=p_{\tau,x}hp_{\sigma,f(x)}.
$$
Substituting $p=\iota\rho$ yields
$$
\begin{aligned}
h
&=
\iota_{\tau,x}\rho_{\tau,x}
 h
\iota_{\sigma,f(x)}\rho_{\sigma,f(x)}\\
&=
\iota_{\tau,x}
Df(x)^{\mathsf T}
\rho_{\sigma,f(x)}.
\end{aligned}
$$
The same calculation proves uniqueness among ambient maps satisfying the
restriction and sandwich equations.

For identities, $h=\iota\rho=p$. For a composite, the adjacent retraction and
section cancel because $\rho\iota=1$. The semifunctor laws are therefore the
transport of ordinary derivative functoriality through the retract
presentations.

\subsection{A branch-sensitive calculation}

Let
$$
t:\mathbb R+(\mathbb R\times\mathbb R)\to\mathbb R,
\qquad
t(\mathsf{inl}(x))=x^2,
\qquad
t(\mathsf{inr}(y,z))=yz.
$$
Choose the ambient cotangent $\mathbb R^3$. The left fibre is split by
$$
\iota_L(a)=(a,0,0),
\qquad
\rho_L(a,b,c)=a,
$$
and the right fibre by
$$
\iota_R(b,c)=(0,b,c),
\qquad
\rho_R(a,b,c)=(b,c).
$$
The associated projectors are $p_L=\iota_L\rho_L$ and
$p_R=\iota_R\rho_R$.

On the left branch,
$$
D(x\mapsto x^2)(x)^{\mathsf T}(\bar r)=2x\bar r,
$$
so
$$
\widehat R_t(\mathsf{inl}(x))(\bar r)
=
\iota_L(2x\bar r)
=(2x\bar r,0,0).
$$
On the right branch,
$$
D((y,z)\mapsto yz)(y,z)^{\mathsf T}(\bar r)
=(z\bar r,y\bar r),
$$
and hence
$$
\widehat R_t(\mathsf{inr}(y,z))(\bar r)
=
\iota_R(z\bar r,y\bar r)
=(0,z\bar r,y\bar r).
$$
Thus the ambient result is the mathematical derivative on the active fibre,
followed by its inclusion into the common cotangent host.

\section{The forward-mode dual}
\label{app:forward}

The paper is phrased in reverse mode because backpropagation is the practically dominant application and because the cotangent interpretation makes the implementation problem especially familiar. Semantically, however, the construction is not specific to reverse mode. From the outset, CHAD presents forward and reverse differentiation through fibrewise-dual total categories~\cite{Vakar2021,VakarSmeding2022,LucatelliNunesVakar2023}. The present construction therefore applies to forward mode by choosing the Grothendieck construction before, rather than after, taking fibrewise opposites.

Let
$$
\TotTarFwd=\Sigma_{\CTar}\LTar
$$
be the ordinary forward target. Its concrete constant-tangent semantics is $\Copow(\Vect)$, while the dependent semantics is $\Fam(\Vect)$. Since $\Vect$ is Cauchy complete and direct sums provide common retract hosts, \Cref{thm:criterion} gives
$$
\Kar\bigl(\Copow_\kappa(\Vect)\bigr)
\simeq
\Fam_\kappa(\Vect).
$$

\begin{corollary}[Forward-mode completion]
\label{thm:forward-app}
Under the corresponding forward-target hypotheses of basic CHAD~\cite[Theorem~6.1]{VakarSmeding2022}, and assuming that $\CTar$ has finite coproducts, assigning the chosen primitive derivatives extends, relative to the selected structures, to a bicartesian closed functor
$$
\DFKar:\Synp\longrightarrow\Kar(\TotTarFwd).
$$
Forgetting the distinguished idempotents yields a semifunctor
$$
\DFst:\Synp\semiarrow\TotTarFwd
$$
that preserves composition strictly and sends identities to the corresponding full tangent idempotents. At data types, their primal components are identities and their linear components are the tangent projectors.

If $t:\tau\to\sigma$ has data-type domain and codomain, denotes a map $f$
differentiable at $x$, and
$$
T_\tau(x)
\xrightarrow{\iota_{\tau,x}}
\widehat T_\tau
\xrightarrow{\rho_{\tau,x}}
T_\tau(x)
$$
and
$$
T_\sigma(f(x))
\xrightarrow{\iota_{\sigma,f(x)}}
\widehat T_\sigma
\xrightarrow{\rho_{\sigma,f(x)}}
T_\sigma(f(x))
$$
are the fixed-point splittings determined by the comparison with dependent
forward CHAD, then the complete ambient forward derivative is
$$
\widehat F_t(x)
=
\iota_{\sigma,f(x)}
\circ Df(x)
\circ\rho_{\tau,x}.
$$
It is the unique linear map $h:\widehat T_\tau\to\widehat T_\sigma$ satisfying
$$
h=p_{\sigma,f(x)}hp_{\tau,x},
\qquad
\rho_{\sigma,f(x)}h\iota_{\tau,x}=Df(x).
$$
\end{corollary}

\begin{proof}
Replace the fibrewise-opposite construction $\Sigma_{\CTar}\LTar^{\mathrm{op}}$ by $\Sigma_{\CTar}\LTar$. At the object level, the completion theorem and the projector presentation of variants use only primal-indexed idempotents together with finite coproducts and biproducts, and hence carry over unchanged. At the arrow level, removing the fibrewise opposite reverses the variance of the linear maps and exchanges the operational roles of the biproduct injections and projections. Under the corresponding forward-target closure assumptions, the universal property of $\Synp$ gives the stated bicartesian closed functor; forgetting the distinguished idempotents gives the semifunctor as in \Cref{prop:semi}.

Under the equivalence
$$
\Kar\bigl(\Copow(\Vect)\bigr)
\simeq
\Fam(\Vect),
$$
the restriction of the ambient tangent map to the splitting objects is the ordinary derivative:
$$
\rho_{\sigma,f(x)}h\iota_{\tau,x}=Df(x).
$$
Karoubi compatibility gives $h=p_{\sigma,f(x)}hp_{\tau,x}$. Substituting
$$
p_{\sigma,f(x)}
=
\iota_{\sigma,f(x)}\rho_{\sigma,f(x)},
\qquad
p_{\tau,x}
=
\iota_{\tau,x}\rho_{\tau,x}
$$
yields
$$
h
=
\iota_{\sigma,f(x)}
Df(x)
\rho_{\tau,x}.
$$
The same calculation proves uniqueness. Thus the forward and reverse transformations are fibrewise-dual presentations of the same completion principle; the transpose appears only in the reverse concrete semantics.
\end{proof}
\newpage

\bibliographystyle{splncs04}
\bibliography{references}

\end{document}